\documentclass[journal, 12pt, draftclsnofoot, onecolumn, twoside]{IEEEtran}

\usepackage{mathtools, amsthm}
\usepackage{amssymb}
\usepackage{bbm}
\usepackage{mathtools}
\usepackage{amsfonts}
\usepackage{subcaption}
\usepackage{amssymb}
\usepackage{graphicx}
\usepackage{color}
\usepackage{mathrsfs}
\usepackage{url}
\usepackage{xcolor}
\usepackage{setspace}
\usepackage{float}
\usepackage{hyperref}
\usepackage{url}
\usepackage{xcolor}
\usepackage{cancel}
\usepackage{cite}
\usepackage{soul}
\newtheorem{theorem}{Theorem}
\allowdisplaybreaks

\makeatletter
\def\endthebibliography{%
	\def\@noitemerr{\@latex@warning{Empty `thebibliography' environment}}%
	\endlist
}
\makeatother

\DeclareMathOperator*{\argmax}{arg\,max}
\DeclareMathOperator*{\argmin}{arg\,min}

\newcommand{\bx}{\mathbf{x}}

\newcommand{\bZ}{\mathbf{Z}}

\newcommand{\E}{\mathbbm{E}}

\usepackage{lmodern}
\usepackage{newtxtext}
\usepackage[T1]{fontenc}
\usepackage[utf8]{inputenc}

\title{ Beam Tracking with Photon-Counting Detector Arrays in Free-Space Optical Communications} 
\author{Muhammad~Salman~Bashir,~\emph{Member,~IEEE,}
	and~Mohamed-Slim~Alouini,~\emph{Fellow,~IEEE}
	\thanks{This work is supported by Office of Sponsored Research (OSR) at King Abdullah University of Science and Technology (KAUST). \newline
		M.~S.~Bashir and M.~-S.~Alouini  are with the King Abdullah University of Science and Technology (KAUST), Thuwal, Kingdom of Saudi Arabia 23955-6900.  E-mail: (muhammad.bashir@fulbrightmail.org, slim.alouini@kaust.edu.sa).}
	\thanks{}
	\thanks{}}

\begin{document}
	
	\maketitle
	
	\begin{abstract}
		Optical beam center position on an array of detectors is an important (hidden) parameter that is essential not only from a tracking perspective, but is also important for optimal detection of Pulse Position Modulation symbols in free-space optical communications. In this paper, we have examined the beam position estimation problem for photon-counting detector arrays, and to this end, we have proposed and analyzed a number of non-Bayesian beam position estimators. These estimators are compared in terms of their mean-square error, bias and the probability of error performance. Additionally, the Cram\`er-Rao Lower Bounds (CRLB) of the tracking error is also derived, and the CRLB curves give us additional insights concerning the effect of number of detectors and the beam radius on mean-square error performance. Finally, the effect of beam position estimation on the probability of error performance is investigated, and our study concludes that the probability of error of the system is minimized when the beam position on the array is estimated as accurately as possible. 
	\end{abstract}

\IEEEpeerreviewmaketitle
	
	\section{Introduction}
	Free-space optics (FSO) plays an important role in backhaul networks in 5G wireless communications due to the availability of large chunks of bandwidth in the optical spectrum. However, the problem of pointing, acquisition and tracking is significant in the context of FSO because of the narrow beam widths associated with the optical signal. Acquisition is the process in which the two terminals acquire the initial location of each other before the actual data communication begins. However, after the acquisition is achieved, the system still needs to maintain the alignment between the transmitter and receiver assemblies due to physical factors, such as random effects associated with atmospheric turbulence, the mechanical jitter introduced in the transmitter/receiver assemblies, or building sways due to wind vibrations. This misalignment leads to the loss of received signal energy at the receiver that may increase the outage probability at the receiver. If the beam center can be tracked efficiently on the array, a feed control loop, based on the more agile gimbal-less MEMS retroreflective system \cite{Deng}, can adjust the transmitter/detector assemblies in order to point the field-of-view in the required direction efficiently.

	The interest in pointing and tracking subsystems in FSO has picked up recently due to the deployment of Facebook Connectivity and Google Loon projects \cite{Kaymak} in order to provide internet access to regions of the world that lack a traditional communications infrastructure.  For instance, it is planned that the optical signal from the transmitter will be relayed over to the people in a remote/inaccessible location via a network of balloons/drones. The tracking problem becomes more significant due to the movement of balloons or drones owing to wind motion or inaccurate hovering. Additionally, the demand on accurate tracking becomes more stringent with orbital-angular-momentum beams \cite{Li}.
	
	In this paper, we consider the optical beam position tracking\footnote{Typically, the word ``tracking'' is associated with the ``filtering'' phenomenon where small variations in the parameter are tracked continuously, and all the past data is fused to arrive at the current estimate. In this paper, we only use the current or present set of data to estimate the beam position at each instant of time. }  problem for a free-space optical communication system that employs multiple photon-counting detectors (array of detectors) instead of one large (monolithic) detector at the receiver.  In this study, the purpose of the detector array is twofold: it is used for symbol detection as well as for tracking the beam. Thus, the proposed system is more efficient in terms of bandwidth and hardware complexity since no pilot symbols and extra hardware (mirrors/quadrant photodetectors etc.) are needed in order to track the beam. Hence, in light of this argument, the beam position on the array has two roles to serve; i) it provides error signal to the feedback loop in order to adjust the transceiver alignment assemblies, ii) and the beam position on the array is also part of the channel state information needed for optimal detection of data symbols\footnote{Such data symbols may correspond to Pulse position modulation symbols in an \emph{intensity modulated direct detection} (IM/DD) setting.}. 
	
	Since the beam position is unknown, and a possibly random parameter, we have to estimate it in real-time. In this regard, a number of beam position estimators are proposed, and their performance is analyzed. We will see (through the Cram\'er-Rao Lower Bound) that the mean-square performance improves as the number of detectors in the array is increased while keeping the array area fixed. Additionally, the probability of error performance of our system also improves if we increase the number of detectors in the array. However, the improvement in performance comes with the increased overhead of complexity (computational complexity of  estimators and the circuit/storage complexity), as the number of elements in the array is increased. All these ideas will be discussed in the remaining sections of this paper.

	\begin{figure}
		\begin{center}
			\includegraphics[scale=1.5]{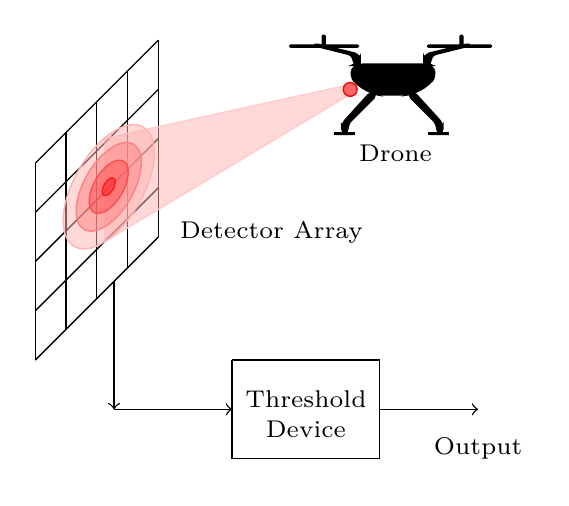}
		\end{center}
	\vspace{-0.2cm}
	\caption{\small A drone projecting the Gaussian beam on a $4 \times 4$ detector array on a ground-based optical receiver.}
	\end{figure}

\section{State-of-the-Art and Contributions of this Paper}

There is no dearth of literature on research in pointing, acquisition and tracking (PAT) systems in FSO that treats the tracking problem purely from a hardware point-of-view. In this respect, \cite{Kaymak} provides a detailed overview of the current state-of-the-art hardware solutions for tracking the optical beam.  Thus, we will cover the literature review from the theoretical/ signal processing perspective since such a perspective is more relevant to our study in this paper. 

The authors in \cite{Gagliardi} and \cite{Marola} have discussed the performance of a proposed feedback (beam) tracking loop that acts on the error signal provided by a quadrant photodetector in the receiver assembly. The work in \cite{Marola} actually builds on the study present in \cite{Gagliardi} by carrying out the stability analysis of their proposed cooperative feedback loop. The authors in \cite{Slocumb} present the performance analysis of centroid and maximum likelihood estimators of beam position for a ``continuous''\footnote{A continuous array is obtained if the number of cells in the detector array goes to infinity while keeping the array area finite. In other words, we have perfect information about the location of each photodetection in the array. Thus, the continuous arrays lead to the best mean-square error performance.} array. Regarding the literature that covers communications with detector arrays in free-space optics, the authors in \cite{Bashir1} propose beam position estimation algorithms  and examine their mean-square error performance with simulations. The work in \cite{Bashir2} extends the work in \cite{Bashir1} by introducing Bayesian filtering algorithms, such as Kalman and particle filters, for tracking the time-varying beam position. The authors in \cite{Bashir3} inspect the relationship between the probability of error and the estimation of beam position on the detector array, and by an argument based on Chernoff Bounds, they show that precise estimation of beam center on the array is necessary in order to minimize the probability of error. Additionally, the author in \cite{Bashir4} presents a mathematical argument to show that the probability of error decreases monotonically as the number of cells in the array is increased. Furthermore, the authors in \cite{Bashir6} analyze the acquisition performance of an FSO system that employs an array of detectors at the receiver. Finally, the authors in \cite{Bashir5} consider time synchronization schemes based on an array of detectors.

Furthermore, we also briefly discuss the literature on pointing and tracking in FSO systems that examine the tracking problem from the perspective of a single detector. In this regard, \cite{Farid} develops the pointing error statistics for a circularly shaped detector and a Gaussian beam, and the outage capacity is optimized as a function of beam radius. The authors in \cite{Mai} investigate a slightly different optimization problem concerning pointing: The maximization of link availability as a function of beam radius (for fixed signal power). Additionally, they also explore the minimization of transmitted power by tuning to the optimal beam radius under the constraint of a fixed link availability. In addition to these papers, the interested reader may be directed to \cite{Ansari, Zedini, Quwaiee, Issaid} for a detailed study on the performance of FSO systems when the optical channel suffers degradation due to pointing errors for a single detector receiver.

This paper proposes a number of  beam tracking algorithms for FSO communications with an array of detectors and a Gaussian beam.  As discussed in the introduction, we plan to analyze the joint problem of beam tracking and data detection with an array of detectors in order to skimp on the required bandwidth and energy for our system \cite{Ferraro:15}. Hence, in our decision-directed scheme, the data symbols aid the beam tracking on the array, and, in return, the efficient beam tracking process helps with accurate detection of symbols. In this regard, we build on the study done in \cite{Bashir1} by providing a more theoretical framework for the tracking problem. Hence, in addition to proposing a number of additional estimators, we analyze the mean-square tracking error in terms of Cram\'er-Rao Lower Bound, which gives us some deeper insights into the tracking performance of the system. Furthermore, we explore a few interesting asymptotic scenarios which simplify the expressions, and help us get a better understanding of the problem. The algorithms proposed in this paper use photon counts in the detectors generated during an observation interval as a sufficient statistic for tracking the beam position. Photon-counting detectors provide a better probability of error performance as compared to analog detectors for low signal-to-noise ratio scenario \cite{Bashir6}.

 Additionally, the effect of the beam tracking algorithms on the probability of error is also analyzed. We reason---by using an analytical argument for the asymptotic case scenario (infinite number of detectors and poor SNR)---that the probability of error is minimized when beam position on the array is estimated accurately enough. Even though, the authors in \cite{Bashir3} have presented an argument on the minimization of probability of error (as a function of beam position) using Chernoff Bounds, the arguments presented in this study are more robust. 

The major assumption regarding tracking with detector arrays is that the array area is large enough so that the  beam footprint is smaller than the footprint of the array. This is a valid assumption for channels which are not marred by scintillation effects due to turbulence (e.g., an optical link in the stratosphere), or for channels where the length of the link is of the order of a few hundred kilometers \cite{Bashir6}. Additionally, in this paper, the focus is on non-Bayesian estimation techniques for beam tracking. This is due to the fact that unless we are certain about the parameters of the prior random motion model of the beam on the array, we are likely going to incur a significant loss in performance if there is mismatch in our assumptions and the real world parameters\footnote{This is especially true if the parameters themselves---such as the covariance matrices of the random motion model---are time-varying.} \cite{Shmaliy}.

This paper is organized as follows. In Section~\ref{system_model}, we discuss the system model, and Section~\ref{CRLB} contains the derivation of the Cram\'er-Rao Lower Bound for the beam position tracking error. The high complexity trackers, such as nonlinear least squares estimator and maximum likelihood estimators are described in Section~\ref{HCT}, and the maximum detector count estimator and different versions of the centroid estimator are discussed in Section~\ref{LCT}. Section~\ref{PE} considers the probability of error analysis for different beam position trackers, and Section~\ref{sims} elaborates on simulation principles and parameters. This is followed by a brief complexity analysis and the conclusions which are summarized in Section~\ref{Complexity} and Section~\ref{Conc}, respectively.

\section{System Model for Beam Position Tracking} \label{system_model}
The received optical signal on the receiver aperture gives rise to photoelectrons or photodetections in each detector of the array due to the \emph{photoelectric effect}. The emission of these photoelectrons during the signal pulse interval helps us detect transmitted symbols. The photon count $Z_m$ in the $m$th detector or cell of the array---during some observation interval $T_s$---is modeled as a discrete random variable. Its probability mass function  is characterized by the following expression:
\begin{align}
P(\{ Z_m = z_m \}) = \frac{ \exp \left( {-\iint_{A_m} \left[ \lambda_s(x,y) + \lambda_n \right] \, dx\, dy} \right) (\iint_{A_m} \left[ \lambda_s(x,y) + \lambda_n \right] \, dx\, dy)^{z_m} } {z_m!} \label{photon_model} , \; m = 1, \dotsc, M, 
\end{align}
where $\lambda_s(x, y)$ is the scaled beam intensity\footnote{The actual signal intensity, $\lambda_{s_i}$, and the actual noise intensity, $\lambda_{n_i}$, are multiplied by the constant $ \frac{\eta T_s}{hc/\lambda}$ in order to obtain the intensity $\lambda_s$ and $\lambda_n$ for the photon generation model in \eqref{photon_model}. The constant $h$ is known as the \emph{Planck's constant}, and its value is $6.62607004 \times 10^{-34}\, {m}^2 kg / s$. The constant $c$ is the speed of light in vacuum which is about $3\times 10^{8} \, m/s$, $\lambda$ is the wavelength of light in meters, $\eta$ stands for the photoconversion efficiency, and $T_p$ represents signal pulse duration.} profile on the detector array, $\lambda_n$ is the scaled noise intensity profile, $A_m$ is the region of the $m$th detector on the detector array, $Z_1, Z_2, \dots, Z_M$ are independent Poisson random variables and $M$ is the total number of detectors in the array. 
\begin{figure}
	\begin{center}
		\includegraphics[scale=0.5]{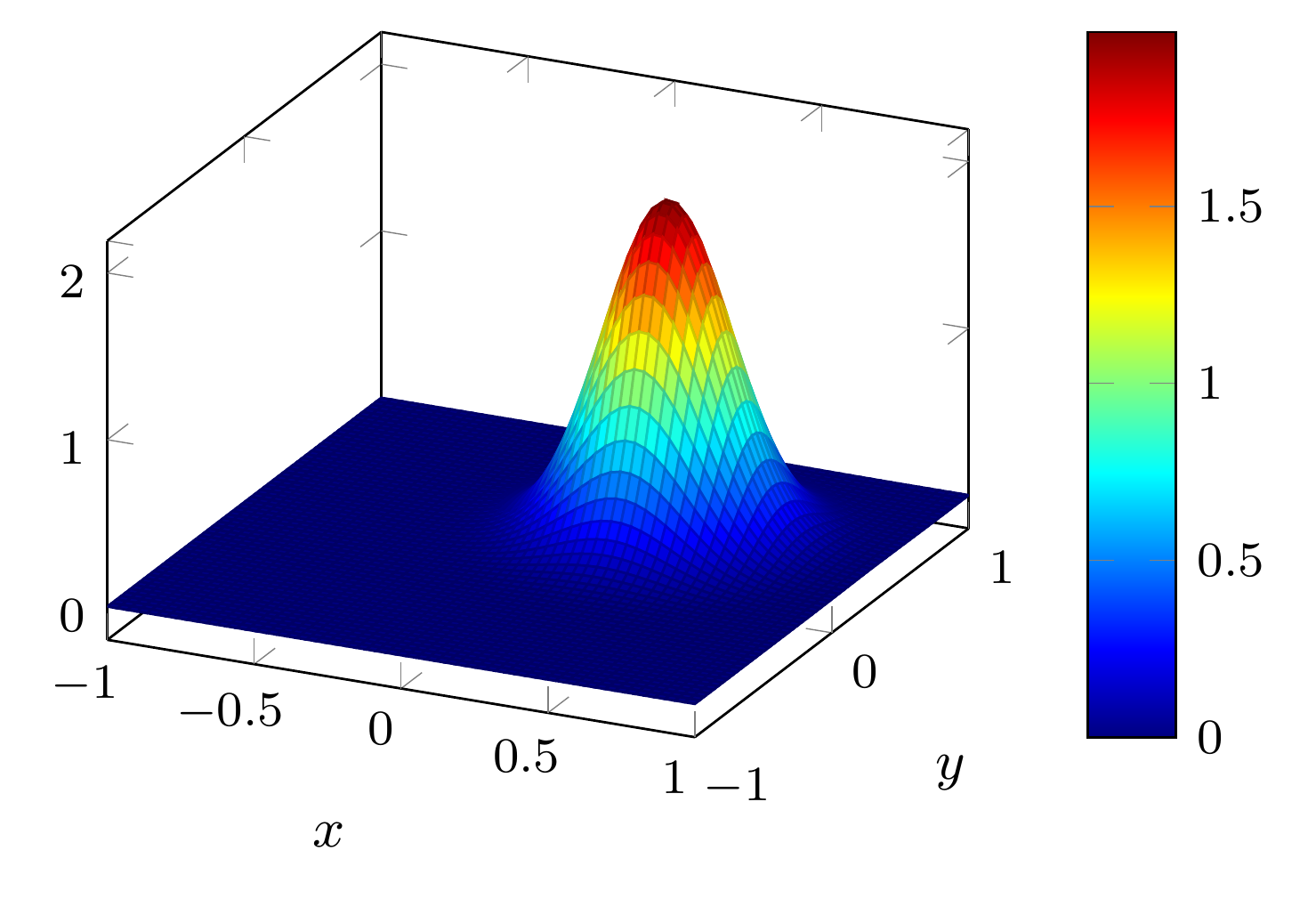}
	\end{center}
	\caption{\small Profile of the incident beam on the detector array.}
\end{figure} As may have been discerned by the reader, the coordinate $(x,y)$ stands for any point inside the region of the detector array. Moreover, $\lambda_n$ is a constant factor that accounts for the background radiation and the thermal effects of the detector array \cite{Snyder}.

For Gaussian beams, the received (scaled) signal and noise intensity at the detector array is given by the expression
\begin{align}
\lambda_s(x,y,d)& \triangleq \frac{I_0}{\rho^2(d)} \exp\left( \frac{-(x-x_0)^2 -(y-y_0)^2}{2 \rho^2(d)}  \right), \label{intensity}
\end{align}
where $I_0/\rho^2(d)$ is the peak intensity in W/$\text{m}^2$/s,  $\rho(d) = \rho_0 \sqrt{1 + \left(\frac{\lambda d}{\pi \rho_0^2} \right)^2}$ meters, and $(x_0, y_0)$ is the center of the Gaussian beam on the detector array. The factor $\rho_0$ is the \emph{beam waist} measure in meters, and $\rho(d)$ is known at the \emph{beam radius} or the \emph{spot size} at a distance $d$ meters from the transmitter.
Finally, the constant $\lambda_n$ denotes the uniformly distributed background radiation intensity and noise effect at the receiver. 

In the latter sections, the dependence of $\lambda(x,y,d)$ and $\rho(d)$ on distance $d$ is removed since we assume that $d$ is fixed, and therefore, $\lambda(x,y) \triangleq \lambda(x,y,d)$ and $\rho \triangleq \rho(d)$. Furthermore, it is a general assumption in the following sections that the center of the array has the coordinates $(0,0)$, and that the array extends from $-a$ to $a$ in each dimension for $a\in \mathbb{R}^{+}$.  Additionally, the area of $A_m$ is denoted by $A$ since all detectors are assumed to have an equal area.

\section{Cram\'er-Rao Lower Bound for Beam Position Tracking Error} \label{CRLB}
In this section, we derive the \emph{Cram\'er-Rao Lower Bound} (CRLB) for the beam position tracking error. In this regard, the likelihood function is given by
\begin{align}
p(\bZ|x_0, y_0) = \prod_{m=1}^M {e^{-\Lambda_m}} \frac{\Lambda_m^{z_m}}{z_m!},
\end{align} 
where
 \begin{align}
\Lambda_m \triangleq \iint_{A_m} \left( \frac{I_0}{\rho^2} e^{-\frac{(x-x_0)^2 + (y-y_0)^2}{2\rho^2}} + \lambda_n\right)\,dx \,dy, \label{intensity1}
\end{align} 
and the random vector $\bZ \triangleq \begin{bmatrix}
Z_1 & Z_2 & \dotsb & Z_M
\end{bmatrix}^T$. Let us define the total incident power on the array $\Lambda_s \triangleq \sum_{m=1}^M \Lambda_m$. Then, 
\begin{align}
\ln p(\bZ|x_0, y_0) = \sum_{m=1}^M z_m \ln \Lambda_m - \Lambda_m - \ln z_m! = \sum_{m=1}^M z_m \ln \Lambda_m -\ln z_m! - \Lambda_s.
\end{align}
Thus,
\begin{align}
&\frac{\partial  \ln p(\bZ|x_0, y_0)}{\partial x_0} = \sum_{m=1}^M \frac{Z_m}{\Lambda_m} \iint_{A_m} \frac{I_0}{\rho^2} e^{-\frac{(x-x_0)^2 + (y-y_0)^2}{2 \rho^2}} \frac{(x-x_0)}{\rho^2}\, dx \,dy - \underbrace{\iint_{\mathcal{A}} \frac{I_0}{\rho^2} e^{-\frac{(x-x_0)^2 + (y-y_0)^2}{2 \rho^2}} \frac{(x-x_0)}{\rho^2}\, dx \,dy}_{0}\nonumber \\
&= \sum_{m=1}^M \frac{Z_m}{\Lambda_m} \iint_{A_m} \frac{I_0}{\rho^4} (x-x_0) e^{-\frac{(x-x_0)^2 + (y-y_0)^2}{2 \rho^2}} \, dx \,dy 
\end{align}
and 
\begin{align}
&\frac{\partial^2 \ln p(\bZ|x_0, y_0)}{\partial x_0^2} = \sum_{m=1}^M -\frac{Z_m}{\Lambda_m^2} \left( \iint_{A_m} \frac{I_0}{\rho^4} (x-x_0) e^{-\frac{(x-x_0)^2 + (y-y_0)^2}{2 \rho^2}} \, dx \,dy \right)^2 \nonumber \\
& -\sum_{m=1}^M \frac{Z_m}{\Lambda_m} \iint_{A_m} \frac{I_0}{\rho^4} e^{-\frac{(x-x_0)^2 + (y-y_0)^2}{2 \rho^2}} \, dx \,dy + \sum_{m=1}^M \frac{Z_m}{\Lambda_m} \iint_{A_m} \frac{I_0}{\rho^6} (x-x_0)^2 e^{-\frac{(x-x_0)^2 + (y-y_0)^2}{2 \rho^2}} \, dx \,dy.
\end{align}
Now, the expectation is taken with respect to $Z_m$:
\begin{align}
&-\E\left[ \frac{\partial^2\ln p(\bZ|x_0, y_0)}{\partial x_0^2 }  \right] = \sum_{m=1}^M \frac{1}{\Lambda_m} \left( \iint_{A_m} \frac{I_0}{\rho^4} (x-x_0) e^{-\frac{(x-x_0)^2 + (y-y_0)^2}{2 \rho^2}} \, dx \,dy \right)^2 \nonumber \\
& +\sum_{m=1}^M  \iint_{A_m} \frac{I_0}{\rho^4} e^{-\frac{(x-x_0)^2 + (y-y_0)^2}{2 \rho^2}} \, dx \,dy - \sum_{m=1}^M  \iint_{A_m} \frac{I_0}{\rho^6} (x-x_0)^2 e^{-\frac{(x-x_0)^2 + (y-y_0)^2}{2 \rho^2}} \, dx \,dy \nonumber \\
&= \sum_{m=1}^M \frac{1}{\Lambda_m} \left( \iint_{A_m} \frac{I_0}{\rho^4} (x-x_0) e^{-\frac{(x-x_0)^2 + (y-y_0)^2}{2 \rho^2}} \, dx \,dy \right)^2 \nonumber \\
& + \iint_{\mathcal{A}} \frac{I_0}{\rho^4} e^{-\frac{(x-x_0)^2 + (y-y_0)^2}{2 \rho^2}} \, dx \,dy -  \iint_{\mathcal{A}} \frac{I_0}{\rho^6} (x-x_0)^2 e^{-\frac{(x-x_0)^2 + (y-y_0)^2}{2 \rho^2}} \, dx \,dy \nonumber \\
&= \sum_{m=1}^M \frac{1}{\Lambda_m} \left( \iint_{A_m} \frac{I_0}{\rho^4} (x-x_0) e^{-\frac{(x-x_0)^2 + (y-y_0)^2}{2 \rho^2}} \, dx \,dy \right)^2 + \frac{2\pi I_0}{\rho^2} - \frac{2 \pi I_0}{\rho^2} \nonumber \\
&= \sum_{m=1}^M \frac{1}{\Lambda_m} \left( \iint_{A_m} \frac{I_0}{\rho^4} (x-x_0) e^{-\frac{(x-x_0)^2 + (y-y_0)^2}{2 \rho^2}} \, dx \,dy \right)^2.
\end{align}
Similarly, it can be shown that 
\begin{align}
-\E\left[ \frac{\partial^2\ln p(\bZ|x_0, y_0)}{\partial y_0^2 }  \right]& = \sum_{m=1}^M \frac{1}{\Lambda_m} \left( \iint_{A_m} \frac{I_0}{\rho^4} (y-y_0) e^{-\frac{(x-x_0)^2 + (y-y_0)^2}{2 \rho^2}} \, dx \,dy \right)^2.
\end{align}

Furthermore, 
\begin{align}
&\frac{\partial^2 \ln p(\bZ|x_0, y_0)}{\partial x_0 \partial y_0} = \sum_{m=1}^M -\frac{Z_m}{\Lambda_m^2} \iint_{A_m} \frac{I_0}{\rho^4} (y-y_0) e^{-\frac{(x-x_0)^2 + (y-y_0)^2}{2 \rho^2}} \, dx \,dy\times \iint_{A_m} \frac{I_0}{\rho^4} (x-x_0) e^{-\frac{(x-x_0)^2 + (y-y_0)^2}{2 \rho^2}} \, dx \,dy \nonumber \\
&+ \sum_{m=1}^M\frac{Z_m}{\Lambda_m} \iint_{A_m} \frac{I_0}{\rho^6} (x-x_0) (y-y_0) e^{-\frac{(x-x_0)^2 + (y-y_0)^2}{2 \rho^2}} \, dx \,dy 
\end{align}
\begin{align}
&-\E\left[ \frac{\partial^2\ln p(\bZ|x_0, y_0)}{\partial x_0 \partial y_0 }  \right] \nonumber \\
&= \sum_{m=1}^M \frac{1}{\Lambda_m} \iint_{A_m} \frac{I_0}{\rho^4} (y-y_0) e^{-\frac{(x-x_0)^2 + (y-y_0)^2}{2 \rho^2}} \, dx \,dy \iint_{A_m} \frac{I_0}{\rho^4} (x-x_0) e^{-\frac{(x-x_0)^2 + (y-y_0)^2}{2 \rho^2}} \, dx \,dy \nonumber \\
&-\underbrace{\iint_{\mathcal{A}} \frac{I_0}{\rho^6} (x-x_0) (y-y_0) e^{-\frac{(x-x_0)^2 + (y-y_0)^2}{2 \rho^2}} \, dx \,dy}_{0} \nonumber \\
&= \sum_{m=1}^M \frac{1}{\Lambda_m} \iint_{A_m} \frac{I_0}{\rho^4} (y-y_0) e^{-\frac{(x-x_0)^2 + (y-y_0)^2}{2 \rho^2}} \, dx \,dy \iint_{A_m} \frac{I_0}{\rho^4} (x-x_0) e^{-\frac{(x-x_0)^2 + (y-y_0)^2}{2 \rho^2}} \, dx \,dy\nonumber \\
&= -\E\left[ \frac{\partial^2\ln p(\bZ|x_0, y_0)}{\partial y_0 \partial x_0 }  \right].
\end{align}
 Moreover, the \emph{Fisher Information Matrix} is
 \begin{align}
 I(x_0, y_0) = \begin{bmatrix}
 -\E\left[ \frac{\partial^2\ln p(\bZ|x_0, y_0)}{\partial x_0^2 }  \right] & -\E\left[ \frac{\partial^2\ln p(\bZ|x_0, y_0)}{\partial x_0 \partial y_0 } \right] \\[8pt]
  -\E\left[ \frac{\partial^2\ln p(\bZ|x_0, y_0)}{\partial y_0 \partial x_0 } \right] & -\E\left[ \frac{\partial^2\ln p(\bZ|x_0, y_0)}{\partial y_0^2 }  \right]
 \end{bmatrix},
 \end{align}
	and $\text{Var}\left[ \hat{x}_0\right] \geq \left[I^{-1}(x_0, y_0)\right]_{1,1}$, and $\text{Var}\left[ \hat{y}_0\right] \geq \left[I^{-1}(x_0, y_0)\right]_{2,2}$. 
	Finally, 
	\begin{align}
	&\text{Var}[\hat{x}_0] \geq {\sum_{m=1}^M \frac{1}{\Lambda_m} \left( \iint_{A_m} \frac{I_0}{\rho^4} (y-y_0) e^{-\frac{(x-x_0)^2 + (y-y_0)^2}{2 \rho^2}} \, dx \,dy \right)^2}\nonumber \\
	&\div \left[\sum_{m=1}^M \frac{1}{\Lambda_m} \left( \iint_{A_m} \frac{I_0}{\rho^4} (x-x_0) e^{-\frac{(x-x_0)^2 + (y-y_0)^2}{2 \rho^2}} \, dx \,dy \right)^2 \times \sum_{m=1}^M \frac{1}{\Lambda_m} \left( \iint_{A_m} \frac{I_0}{\rho^4} (y-y_0) e^{-\frac{(x-x_0)^2 + (y-y_0)^2}{2 \rho^2}} \, dx \,dy \right)^2 \right. \nonumber \\ 
	&- \left. \left(\sum_{m=1}^M \frac{1}{\Lambda_m} \iint_{A_m} \frac{I_0}{\rho^4} (y-y_0) e^{-\frac{(x-x_0)^2 + (y-y_0)^2}{2 \rho^2}} \, dx \,dy \iint_{A_m} \frac{I_0}{\rho^4} (x-x_0) e^{-\frac{(x-x_0)^2 + (y-y_0)^2}{2 \rho^2}} \, dx \,dy\right)^2 \right], \label{crlbx}
	\end{align}
	and 
	\begin{align}
	&\text{Var}[\hat{y}_0] \geq {\sum_{m=1}^M \frac{1}{\Lambda_m} \left( \iint_{A_m} \frac{I_0}{\rho^4} (x-x_0) e^{-\frac{(x-x_0)^2 + (y-y_0)^2}{2 \rho^2}} \, dx \,dy \right)^2}\nonumber \\
	&\div \left[\sum_{m=1}^M \frac{1}{\Lambda_m} \left( \iint_{A_m} \frac{I_0}{\rho^4} (x-x_0) e^{-\frac{(x-x_0)^2 + (y-y_0)^2}{2 \rho^2}} \, dx \,dy \right)^2 \times \sum_{m=1}^M \frac{1}{\Lambda_m} \left( \iint_{A_m} \frac{I_0}{\rho^4} (y-y_0) e^{-\frac{(x-x_0)^2 + (y-y_0)^2}{2 \rho^2}} \, dx \,dy \right)^2 \right. \nonumber \\ 
	&- \left. \left(\sum_{m=1}^M \frac{1}{\Lambda_m} \iint_{A_m} \frac{I_0}{\rho^4} (y-y_0) e^{-\frac{(x-x_0)^2 + (y-y_0)^2}{2 \rho^2}} \, dx \,dy \iint_{A_m} \frac{I_0}{\rho^4} (x-x_0) e^{-\frac{(x-x_0)^2 + (y-y_0)^2}{2 \rho^2}} \, dx \,dy\right)^2 \right]. \label{crlby}
	\end{align}
	
	\subsection{Asymptotic Case ($M \to \infty$)}
	In the following analysis, let us analyze the lower bound on the variance of $\hat{x}_0$ only. The same analysis will hold in the case of lower bound on the variance of $\hat{y}_0$ due to the symmetric nature of the Gaussian beam. 
	\subsubsection{High Signal-To-Noise Ratio}
	For high SNR, $\lambda_n A << \iint_{A_m} \frac{I_0}{\rho^2} e^{-\frac{(x-x_0)^2 + (y-y_0)^2}{2 \rho^2}} \, dx \,dy$. Then, $\Lambda_m \approx \iint_{A_m} \frac{I_0}{\rho^2} e^{-\frac{(x-x_0)^2 + (y-y_0)^2}{2 \rho^2}} \, dx \,dy$. When $M \to \infty$, $ 
	\Lambda_m \approx \frac{I_0}{\rho^2} e^{-\frac{(x_m-x_0)^2 + (y_m-y_0)^2}{2 \rho^2}} \Delta_M$, where $(x_m, y_m)$ is the center of the $m$th small cell, and $\Delta_M$ is the infinitesimal area. Then, the numerator of \eqref{crlbx} simplifies as
	\begin{align}
	 &\sum_{m=1}^M \frac{1}{\Lambda_m} \left( \iint_{A_m} \frac{I_0}{\rho^4} (y-y_0) e^{-\frac{(x-x_0)^2 + (y-y_0)^2}{2 \rho^2}} \, dx \,dy\right)^2 \approx \sum_{m=1}^M \frac{\left(\frac{I_0}{\rho^4} (y_m-y_0) e^{-\frac{(x_m-x_0)^2 + (y_m-y_0)^2}{2 \rho^2}} \Delta_M \right)^2} {\frac{I_0}{\rho^2} e^{-\frac{(x_m-x_0)^2 + (y_m-y_0)^2}{2 \rho^2} } \Delta_M} \nonumber \\
	 &= \sum_{m=1}^M \frac{I_0}{\rho^6}(y_m-y_0)^2 e^{-\frac{(x_m-x_0)^2 + (y_m-y_0)^2}{2 \rho^2}} \Delta_M \approx \frac{I_0 2\pi }{\rho^4} \iint_{\mathcal{A}}  \frac{1}{2\pi \rho^2} (y-y_0)^2 e^{-\frac{(x-x_0)^2 + (y-y_0)^2}{2 \rho^2}} \, dx \, dy\nonumber \\
	 &= \frac{I_0 2 \pi}{\rho^4} \rho^2 = \frac{I_0 2 \pi}{\rho^2}.
	\end{align}  
	The positive term in the denominator can be simplified in a similar fashion. The square root of the term with minus sign can be simplified as
	\begin{align}
	&\sum_{m=1}^M \frac{1}{\Lambda_m} \iint_{A_m} \frac{I_0}{\rho^4} (y-y_0) e^{-\frac{(x-x_0)^2 + (y-y_0)^2}{2 \rho^2}} \, dx \,dy \iint_{A_m} \frac{I_0}{\rho^4} (x-x_0) e^{-\frac{(x-x_0)^2 + (y-y_0)^2}{2 \rho^2}} \, dx \,dy\nonumber \\
	&=\approx \sum_{m=1}^M \frac{\frac{I_0}{\rho^4} (y_m-y_0) e^{-\frac{(x_m-x_0)^2 + (y_m-y_0)^2}{2 \rho^2} } \Delta_M  }{\frac{I_0}{\rho^2} e^{-\frac{(x_m-x_0)^2 + (y_m-y_0)^2}{2 \rho^2} } \Delta_M  } \times \frac{I_0}{\rho^4} (x_m-x_0) e^{-\frac{(x_m-x_0)^2 + (y_m-y_0)^2}{2 \rho^2} } \Delta_M \nonumber \\
	&\approx \frac{I_0 2 \pi}{\rho^4}   \iint_{A_m} \frac{1}{2\pi\rho^2} (y-y_0)(x-x_0) e^{-\frac{(x-x_0)^2 + (y-y_0)^2}{2 \rho^2}} \, dx \,dy = 0.
	\end{align}
	Therefore, 
	\begin{align}
	\text{Var}[\hat{x}_0] \geq \frac{\frac{I_0 2 \pi}{\rho^2}} {\frac{I_0 2\pi}{\rho^2} \times \frac{I_0 2\pi}{\rho^2}} = \frac{\rho^2}{I_0 2 \pi}.
	\end{align}
	We note that the CRLB is minimized by minimizing $\rho$ (a more focused beam) for fixed signal power. Additionally, as expected, the CRLB improves with higher $I_0$ (higher signal power).
	\subsubsection{Low Siganl-To-Noise Ratio}
	In this case, let us assume that $\lambda_n A >> \iint_{A_m} \frac{I_0}{\rho^2} e^{-\frac{(x-x_0)^2 + (y-y_0)^2}{2 \rho^2}} \, dx \,dy$. Then, $\Lambda_m \approx \lambda_n A$. In this case, the square root of the term with the minus sign in the denominator is
	\begin{align}
	\frac{1}{\lambda_n A} \sum_{m=1}^M \iint_{A_m} \frac{I_0}{\rho^4} (y-y_0) e^{-\frac{(x-x_0)^2 + (y-y_0)^2}{2 \rho^2}} \, dx \,dy \times \iint_{A_m} \frac{I_0}{\rho^4} (x-x_0) e^{-\frac{(x-x_0)^2 + (y-y_0)^2}{2 \rho^2}} \, dx \,dy,
	\end{align}
	which is zero due to the symmetric nature of the Gaussian beam. Therefore, by further simplification,
	\begin{align}
	\text{Var}[\hat{x}_0] \geq \frac{\lambda_n \rho^8}{I_0^2} \times \frac{\Delta_M}{\sum_{m=1}^M \left( \iint_{A_m} (x-x_0) e^{-\frac{(x-x_0)^2 + (y-y_0)^2}{2 \rho^2}} \, dx \,dy \right)^2}
	\end{align} 
	which goes to \[ \frac{\frac{\lambda_n \rho^8}{I_0^2}}{\iint_{\mathcal{A}}(x-x_0)^2 e^{-\frac{(x-x_0)^2 + (y-y_0)^2}{\rho^2}} \, dx \,dy } = \frac{2  \rho^4}{\pi \left(\frac{I_0^2}{\lambda_n}\right) } \] as $M \to \infty$.
	 In this case, as expected, the CRLB is inversely proportional to $\frac{I_0^2}{\lambda_n}$.
	 
	 \subsection{Analysis of CRLB Curves}\label{CRLB_analysis}
	  Fig.~\ref{fig4a} displays the CRLB curves plotted as a function of noise power. The results indicate a ``diminishing rate of return'' trend as $M$ increases indefinitely. Fig.~\ref{fig4b} depicts the CRLB curves as a function of $\rho$. We can see that there is an optimum value of $\rho$ (lets call it $\rho_M^*$ for the $M$-cell array) at which the CRLB is minimized. Additionally, $\rho^*_{N} < \rho^*_M$ for $N > M$. Intuitively, these observations are straightforward to explain. For  fixed SNR, if the beam footprint is small, but at least covers one cell completely, then such a small beam footprint will minimize the mean-square error. This is true since all the power is focused into a small region on the array where the number of noise photons (on average) is relatively small, and this fact will help the tracker to estimate the beam position more accurately as opposed to a more ``spread out'' beam.  	 
	  
	  However, if the beam radius is much smaller than the dimensions of the cell, then the beam will only give rise to photons in the cell in which it is located, and the neighboring cells will not register any signal photons. Since we round off the locations of the photons---that occur inside a given cell---to the center of the cell, any movement of the ``super thin'' beam inside the given cell cannot be tracked. Therefore, the CRLB rises if $\rho$ diminishes beyond a certain (optimum) value.
	 \begin{figure}
	 	\centering
	 	\hspace{-1cm}
	 	\begin{subfigure}[t]{0.5\textwidth}
	 		\includegraphics[width=\textwidth]{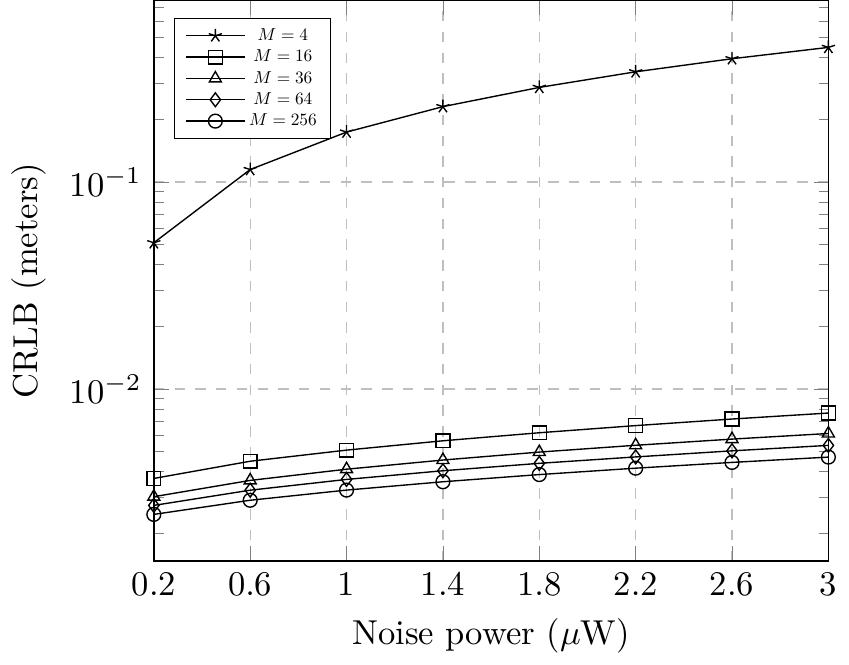} 
	 		\caption{CRLB as a function of noise power} \label{fig4a}
	 	\end{subfigure}
	 	\vspace{-3cm}
	 	\begin{subfigure}[t]{0.5\textwidth}
	 		\includegraphics[width=\textwidth]{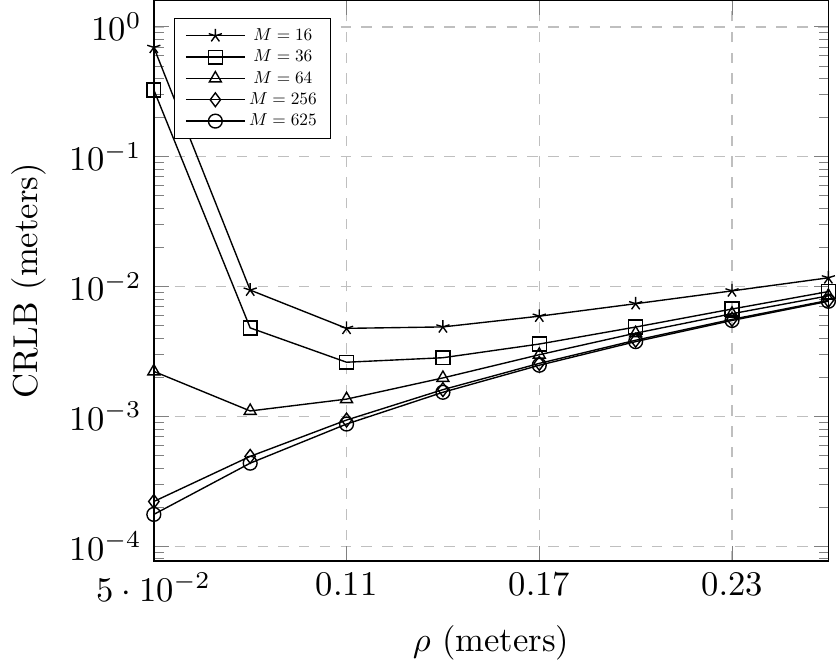} 
	 		\caption{CRLB as a function of $\rho$} \label{fig4b}
	 	\end{subfigure}
	 	\vspace{2.8cm}
	 	\caption{Fig.~\ref{fig4a} depicts the CRLB  of beam position estimators as a function of noise power for different detector arrays. The signal power is 1 $\mu$W, the noise power is varied between 0.2 and 3.0 $\mu$W, and $\rho$ is fixed at 0.2 meters.   Fig.~\ref{fig4b} shows the CRLB plots as a function of beam radius $\rho$. The noise power was fixed at 1.8 $\mu$W in this case. For both the figures, $(x_0, y_0) = (0.15, 0.15)$ and $|\mathcal{A}| = 4$ square meters.} \label{fig4}
	 \end{figure}
	 
	\section{High Complexity Trackers} \label{HCT}
	In this section, we take a look at two high complexity beam position trackers, namely the \emph{nonlinear least squares} (NLS) estimator and the \emph{maximum likelihood estimator} (MLE). As we will see later, these estimators will provide a better mean-square error performance than the low complexity estimators. However, the better performance of these estimators comes with a higher computational complexity, mainly for the following two reasons:
	\begin{enumerate}
		\item Both NLS and MLE are computed as a point in the parameter space where a numerical optimization algorithm converges while maximizing/minimizing a certain objective function (e.g., likelihood function in the case of MLE). An example of such numerical  algorithms is the evolutionary algorithms (genetic algorithm, differential evolution algorithm).
		\item These estimators require estimation of additional beam parameters. In this case, we need to estimate the values of $I_0$, $\rho$ and $\lambda_n$ at the receiver. 
	\end{enumerate}
The next section discusses the estimation of $I_0$ and $\lambda_n$ that is based on the low complexity method of moments estimator.  Unfortunately, estimation of $\rho$ is not so straightforward. For channels that are not marred significantly by turbulence or scattering, $\rho$ can be approximated from the expression $\rho(d) = \rho_0 \sqrt{1 + \left(\frac{\lambda d}{\pi \rho_0^2} \right)^2}$ if the link distance $d$ is known. Otherwise, we will have to estimate it as an unknown parameter alongside the beam center position on the array.

In order to simplify the expressions for the upcoming NLS and MLE estimators, we will replace the Gaussian integrals with the distribution functions of a standard normal random variable. In this regard, we know that the density function of the photon count in the $m$th cell is 
\begin{align}
P\left( \{Z_m = z_m \}  \right) = \frac{e^{-\Lambda_m} \Lambda_m^{z_m}}{z_m!}, \; m=1, \dotsc, M,
\end{align}
where $z_m \in \mathbb{Z}^+ \cup \{0\}$ and $\Lambda_m$ is defined in \eqref{intensity1}.  After a few easy manipulations, $\Lambda_m$ can be simplified as
\begin{align}
\Lambda_m(x_0, y_0) = I_0 2 \pi  \left[  \Phi\left( \frac{y_{m_2} - y_0 }{\rho}  \right) - \Phi\left( \frac{y_{m_1} -y_0}{\rho} \right)   \right] \left[  \Phi\left(  \frac{x_{m_2}-x_0}{\rho} \right) - \Phi\left( \frac{x_{m_1}-x_0}{\rho} \right)   \right] + \lambda_n A,
\end{align}
where $\Phi(x)$ is the distribution function of a standard normal random variable, and \newline $(x_{m_2}, y_{m_2}),  (x_{m_1}, y_{m_2}), (x_{m_2}, y_{m_1}), (x_{m_1}, y_{m_1})$ are the coordinates of the square region $A_m$ such that $x_{m_2} > x_{m_1}$ and $y_{m_2} > y_{m_1}$. 

In the analysis that follows, let us call the center of the $m$th detector $(x_m, y_m)$.

	\subsection{Method of Moments Estimator of $I_0$ and  $\lambda_n$}
	
	
	 In order to use the naive \emph{method of moments} estimator that can estimate  $I_0$ and $\lambda_n$, we send $N$ pulses of signal in $N$ slots of time (``signal+noise'' slots). Moreover, there is another set of $N$ slots in which we do not transmit anything (``noise only'' slots). The width of the pulse and empty slot is the same: $T_p$. Then, by the \emph{strong law of large numbers}, the sample average converges to the true average \emph{almost surely}, and the method of moments estimate of $\lambda_n$ is defined to be
	\begin{align} \label{1}
	\hat{\lambda}_n \triangleq  \frac{1}{|\mathcal{A}| N } \sum_{i=1}^{N} \sum_{m=1}^M z_{m,i}^{(n)},
	\end{align} 
	where $N$ is the total number of observation intervals used for the estimation of $\lambda_n$.
	The quantity $\mathcal{A}$ represents the detector array region, and $|\mathcal{A}| \triangleq \bigcup_{m=1}^M A_m$ represents the total area of the detector array. The random variable $z_{m,i}^{(n)}$ is the noise photon count in the $m$th cell during the $i$th observation interval where these observation intervals correspond to the ``noise only'' slots.
	
	  By the same argument, the method of moments estimate of the signal intensity $I_0$ is
	\begin{align} 
	\hat{I}_0 \triangleq \frac{1} { \Lambda_0 N } \sum_{i=1}^{N} \sum_{m=1}^M\left( z_{m,i}^{(s)} + z_{m,i}^{(n)} \right) - \frac{\hat{\lambda}_n |\mathcal{A}|}{\Lambda_0}. \label{intensity2}
	\end{align}
	 where $\Lambda_0 \triangleq \iint_{\mathcal{A}} \frac{1}{\rho^2} \exp\left(- \frac{(x-x_0)^2+(y-y_0)^2}{2\rho^2} \right)\, dx\, dy$.  The number $z_{m,i}^{(s)}$ corresponds to the signal photons generated in the $m$th detector during the $i$th slot. It is important to know that the count $\sum_{m=1}^M \left( z_{m,i}^{(s)} + z_{m_n}^{(i)} \right)$ in \eqref{intensity2} results assuming that the entire footprint of the beam is captured on the array when the beam center is pointing at the initial estimate of the receiver position.
	 
	 We note that in \eqref{intensity2}, $\Lambda_0$ is a constant with respect to $\rho$. Hence, the estimation of $I_0$ can be carried out independently of the value of $\rho$  at the receiver. 
	
	It is important to note that we have assumed that $I_0$ and $\lambda_n$ remain constant during the duration $NT_p.$ This assumption roughly states that the coherence time of the signal fade is many orders of magnitude larger than $T_p$. This is a fair assumption for high speed data communications in free-space optics where $T_p$ is typically on the order of a fraction of a microsecond.
	
	It can be easily shown that $\hat{I}_0^{(N)}$ and $\hat{\lambda}_n^{(N)}$ are unbiased estimators, and $\hat{I}_0^{(N)} \longrightarrow I_0$ and $\hat{\lambda}_n^{(N)} \longrightarrow \lambda_n$ almost surely as $N \longrightarrow \infty$.

	\subsection{Nonlinear Least Squares Estimator of Beam Position}
	Assuming that $I_0$, $\rho$ and $\lambda_n$ have already been estimated, the NLS estimator of $(x_0, y_0)$ is proposed as follows.
	\begin{align}
	&(\hat{x}_0, \hat{y}_0 ) \nonumber \\
	&\triangleq \argmin_{ x_0, y_0} \sum_{m=1}^{M} \left(z_m\! -\! I_0 2 \pi  \left[  \Phi\left( \frac{y_{m_2} - y_0 }{\rho}  \right) - \Phi\left( \frac{y_{m_1} -y_0}{\rho} \right)   \right]\! \left[  \Phi\left(  \frac{x_{m_2}-x_0}{\rho} \right) - \Phi\left( \frac{x_{m_1}-x_0}{\rho} \right)   \right] \!- \! \lambda_n A \! \right)^2, \label{estimate}
	\end{align}
	 
	\subsection{Maximum Likelihood Estimator of Beam Position}
	The maximum likelihood estimator of beam position on the array is given by \cite{Bashir1} 
	\begin{align}
	&(\hat{x}_0, \hat{y}_0)  \triangleq \argmax_{ x_0, y_0} \ln  p\left(z_1, z_2,\dotsc, z_M|   x_0, y_0 \right) \nonumber \\
	& = \argmax_{ x_0, y_0} \sum_{m=1}^M   z_m \ln \left(  I_0 2 \pi  \left[  \Phi\left( \frac{y_{m_2} - y_0 }{\rho}  \right) - \Phi\left( \frac{y_{m_1} -y_0}{\rho} \right)   \right] \left[  \Phi\left(  \frac{x_{m_2}-x_0}{\rho} \right) - \Phi\left( \frac{x_{m_1}-x_0}{\rho} \right)   \right] + \lambda_n A \right)\nonumber\\
	&-\left( I_0 2 \pi  \left[  \Phi\left( \frac{a - y_0 }{\rho}  \right) - \Phi\left( \frac{-a -y_0}{\rho} \right)   \right] \left[  \Phi\left(  \frac{a-x_0}{\rho} \right) - \Phi\left( \frac{-a-x_0}{\rho} \right)   \right] + \lambda_n |\mathcal{A}| \right) \label{mle}
	\end{align}
		\section{Low Complexity Trackers} \label{LCT}
		In this section, we take a look at a number of low complexity beam position estimators. These estimators are just simple transformations of the photon count vector $\bZ$. Thus, with the exception of the asymptotic unbiased centroid estimator, they do not require the values of beam parameters to compute the estimate.
		
	\subsection{ Maximum Detector Count (MDC) Estimator}
		The \emph{maximum detector count} (MDC) estimator chooses the center of the cell in which the maximum photon count occurs (during some observation interval) as the estimate of the beam position. Let $(x_m, y_m)$ be the center of the $m$th detector, and $\displaystyle n \triangleq \argmax_m Z_m$ for $m=1, \dotsc, M$. Then, the MDC estimator is  defined as 
		\begin{align}
		(\hat{x}_0, \hat{y}_0 ) = (x_{n}, y_{n}).
		\end{align}
		The conditional likelihood of $(\hat{x}_0, \hat{y}_0)$ is bounded as
	\begin{align}
	&p_{\hat{x}_0 \hat{y}_0}(x_m, y_m |x_0, y_0) > P\left( \{ \text{All events such that the photon count in the } m\text{th detector is the greatest}   \} \right) \nonumber \\
	&= P\left( \bigcup_{z_m=1}^\infty \{ Z_m = z_m \} \bigcap \left\{ \bigcap_{\substack {i=1 \\ i\neq m} }^M \{Z_i < z_m \}  \right\}  \right) = \sum_{z=1}^\infty P\left( \{ Z_m =z_m \} \bigcap \left\{  \bigcap_{\substack {i=1 \\ i\neq m} }^M \{Z_i < z_m \}   \right\} \right),  \nonumber \\
	&= \sum_{z_m=1}^\infty P\left( \{ Z_m = z_m \}  \right) \times P\left( \left\{ \bigcap_{\substack {i=1 \\ i\neq m} }^M \{Z_i < z_m \}  \right\} \right)= \sum_{z_m=1}^\infty P\left( \{ Z_m = z_m \}  \right) \times \prod_{\substack {i=1 \\ i\neq m} }^M P\left( \{  Z_i < z_m \}  \right), \nonumber \\
	&= \sum_{z=1}^\infty P\left( \{ Z_m = z_m \}  \right) \times \prod_{\substack {i=1 \\ i\neq m} }^M F_{Z_i}(z_m-1)= \sum_{z_m=1}^\infty \frac{e^{-\Lambda_m} \Lambda_m^z}{z!} \times \prod_{\substack {i=1 \\ i\neq m} }^M \sum_{j=0}^{z_m-1} \frac{e^{-\Lambda_i} \Lambda_i^j}{j!}  \nonumber \\
	&= \underbrace{e^{-\Lambda_m} \sum_{z_m=1}^\infty \left( \frac{\Lambda_m^{z_m}}{z_m!} \times    \prod_{\substack {i=1 \\ i\neq m} }^M Q(z_m, \Lambda_i) \right)}_{P_{1,m}}, \label{max_count}
	\end{align}
	where $Q(z, \Lambda_i)$ is the \emph{regularized Gamma function} and is defined as 
	$
	Q(x, y) \triangleq \frac{\Gamma(x,y)}{\Gamma(x)}, \label{rgamma}
	$
	where $\Gamma(x, y)$ is the \emph{upper incomplete Gamma function}:
	$
	\Gamma(x, y) \triangleq   \int_{y}^{\infty} t^{x-1} e^{-t}\, dt,
	$
	and $ \Gamma(x) \triangleq \int_{0}^{\infty} t^{x-1} e^{-t}\, dt$.
	In the analysis discussed above, we assume that the probability of the  event that two or more detectors report an equal number of (maximum) photon count is small, and the lower bound is tight. However, in case such an event arises, we randomly choose the center of a cell among all cells that report the maximum photon count as our MDC estimator.  The lower bound  in \eqref{max_count} can be improved if we include the possibility of two or more detectors obtaining the same maximum. If we define,
	\begin{align}
P_{2,m} &\triangleq \sum_{z_m=1}^\infty \sum_{\substack{m_1=1 \\ m_1\neq m}}^M  \left( \frac{1}{2}\left( e^{-\Lambda_m}   \frac{\Lambda_m^{z_m}}{z_m!} \times e^{-\Lambda_{m_1}}   \frac{\Lambda_{m_1}^{z_m}}{z_m!} \right) \times    \prod_{\substack {i=1 \\ i\neq m\\ i \neq m_1} }^M Q(z_m, \Lambda_i) \right),\\
P_{3,m}& \triangleq \sum_{z_m=1}^\infty \sum_{\substack{m_1=1 \\ m_1\neq m}}^M \sum_{\substack{m_2=1 \\ m_2 \neq m\\ m_2 \neq m_1}}^M  \left( \frac{1}{3}\left( e^{-\Lambda_m}   \frac{\Lambda_m^{z_m}}{z_m!} \times e^{-\Lambda_{m_1}}   \frac{\Lambda_{m_1}^{z_m}}{z_m!} \times e^{-\Lambda_{m_2}}   \frac{\Lambda_{m_2}^{z_m}}{z_m!} \right) \times    \prod_{\substack {i=1 \\ i\neq m\\ i \neq m_1 \\ i \neq m_2} }^M Q(z_m, \Lambda_i) \right),
	\end{align}
	and for any integer $k$ such that $1<k\leq M$, 
	\begin{align}
	P_{k,m}& \triangleq \sum_{z_m=1}^\infty \sum_{\substack{m_1=1 \\ m_1\neq m}}^M \sum_{\substack{m_2=1 \\ m_2 \neq m\\ m_2 \neq m_1}}^M \dotsm \sum_{\substack{m_{k-1}=1 \\ m_{k-1} \neq m\\ m_{k-1} \neq m_1 \\ \vdots \\ m_{k-1} \neq m_{k-2}}}^M \!\!\left(\! \frac{1}{k}\left( e^{-\Lambda_m}   \frac{\Lambda_m^{z_m}}{z_m!} \times e^{-\Lambda_{m_1}}   \frac{\Lambda_{m_1}^{z_m}}{z_m!} \times \dotsm \times e^{-\Lambda_{m_{k-1}}}   \frac{\Lambda_{m_{k-1}}^{z_m}}{z_m!} \right) \! \times  \!\!\!\!  \prod_{\substack {i=1 \\ i\neq m\\ i \neq m_1\\ \vdots \\ i \neq m_{k-1} } }^M \! \!\! \!Q(z_m, \Lambda_i) \!\right). \label{MDC}
	\end{align}
	Therefore,
	\begin{align}
	p_{\hat{x}_0 \hat{y}_0}(x_m, y_m|x_0, y_0) = \sum_{n=1}^M P_{n,m},
	\end{align}
	where, it should be noted that $P_{n,m}$ is a function of $(x_0, y_0)$ through $\Lambda_m$. The mean-square error is given by 
	\begin{align}
	&\E[(\hat{x}_0 - x_0)^2 + (\hat{y}_0-y_0)^2] = \sum_{m=1}^M \sum_{n=1}^M \left( (x_m-x_0)^2 + (y_m-y_0)^2 \right) P_{n,m}.  \label{mse}
	\end{align}
	We can compute the bias functions $\E[\hat{x}_0-x_0]$ and $\E[\hat{y}_0 - y_0]$ in a similar fashion as \eqref{mse}. For instance, 
	\begin{align}
	\E[\hat{x}_0 - x_0] = \sum_{m=1}^M \sum_{n=1}^M \left( x_m-x_0 \right) P_{n,m}. \label{bias_mdc}
	\end{align}
	 
	\subsection{Centroid Estimator}
	The centroid estimate of the beam position is given by 
	\begin{align}
	\hat{x}_0 &\triangleq  \frac{1}{Z_s}\sum_{m=1}^M x_m Z_m,  \quad \hat{y}_0 &\triangleq  \frac{1}{Z_s}\sum_{m=1}^M y_m Z_m.
	\end{align}
	where $Z_s \triangleq \sum_{m=1}^M Z_m$. The mean-square error of the centroid estimator is derived as follows.
	\begin{align}
	&\E[(\hat{x}_0-x_0)^2 + (\hat{Y}_0-y_0)^2]= \E[(\hat{x}_0-x_0)^2] + \E[(\hat{y}_0 - y_0)^2] \nonumber \\
	&= \sum_{z_s=0}^\infty \left( \E[(\hat{x}_0-x_0)^2| Z_s = z_s] + \E[(\hat{y}_0-y_0)^2| Z_s = z_s] \right) P(\{ Z_s = z_s \}) \label{mse_cent}
	\end{align}
	where \begin{align}
	P(\{  Z_s = z_s \}) = e^{-\Lambda_s} \frac{\Lambda_s^{z_s}}{z_s!}
	\end{align}
	and $\Lambda_s \triangleq \sum_{m=1}^M \Lambda_m$.  Thus,
	\begin{align}
	& \E[(\hat{x}_0 - x_0)^2 + (\hat{y}_0 - y_0)^2| Z_s = z_s] = \E[\hat{x}_0^2 | Z_s = z_s] - 2x_0 \E[\hat{x}_0 | Z_s = z_s] + x_0^2 \nonumber \\
	&+ \E[\hat{y}_0^2 | Z_s = z_s] - 2y_0 \E[\hat{y}_0 | Z_s = z_s] + y_0^2\nonumber \\
	& = \frac{1}{z_s^2} \sum_{m=1}^M \sum_{\substack{n=1 \\ n \neq m} }^M x_m x_n \E[Z_m Z_n| Z_s = z_s] + \frac{1}{z_s^2}\sum_{m=1}^M \E[Z_m^2| Z_s = z_s] - 2x_0 \frac{1}{z_s}\sum_{m=1}^M x_m \E[Z_m | Z_s = z_s] + x_0^2 \nonumber \\
	&+\frac{1}{z_s^2} \sum_{m=1}^M \sum_{\substack{n=1 \\ n \neq m} }^M y_m y_n \E[Z_m Z_n| Z_s = z_s] + \frac{1}{z_s^2} \sum_{m=1}^M \E[Z_m^2| Z_s = z_s] - 2y_0 \frac{1}{z_s}\sum_{m=1}^M y_m \E[Z_m | Z_s = z_s] + y_0^2. \label{cent}
	\end{align}
	In order to compute the first and second order conditional expectations in \eqref{cent}, we note the fact that given $Z_s = z_s$, $Z_1, Z_2, \dotsc, Z_M$ are binomial random variables, and the conditional pmf of $Z_m$ is defined as
	\begin{align}
	P(\{ Z_m = z_m \} | \{ Z_s = z_s \}) = \binom{z_s}{z_m} p_m^{z_m} (1-p_m)^{z_s - z_m}, \quad m=1, \dotsc, M, \label{meanx}
	\end{align}
	where $\displaystyle p_m \triangleq \frac{\Lambda_m}{\Lambda_s}$. Moreover, conditioned on the fact that $Z_s = z_s,$ any pair of random variables $Z_m$ and $Z_n$ for $m \neq n$ are not independent. Therefore, it may not be the case that
	\begin{align}
	 \E[Z_m Z_n | Z_s = z_s] \neq \E[Z_m | Z_s = z_s] \E[Z_n | Z_s = z_s].
	\end{align} 
	However, it can be shown that
	\begin{align}
	P( \{Z_m = z_m,  Z_n = z_n \}| \{Z_s = z_s\} )= \binom{z_s}{z_m} p_m^{z_m} \binom{z_s -z_m}{z_n} p_n^{z_n} p_r^{z_s - z_m -z_n}, \; z_m \neq z_n, z_m + z_n \leq z_s,
	\end{align}
	where $p_r$ corresponds to the probability that a photodetection occurs in the region $A_r \triangleq \mathcal{A} - A_m - A_n$, and is defined as $\displaystyle p_r \triangleq  \left({\sum_{ \substack{ i = 1 \\ i \neq m \\ i \neq n}}^M \Lambda_i } \right) / {\Lambda_s}$. The joint expectation is given by
	\begin{align}
	\E[Z_m Z_n | Z_s = z_s] = \sum_{z_m = 0}^{z_s} \sum_{z_n =0}^{z_s - z_m}   z_m z_n \binom{z_s}{z_m} \binom{z_s -z_m}{z_n} p_m^{z_m}  p_n^{z_n} p_r^{z_s - z_m -z_n}, \; z_m \neq z_n, z_m + z_n \leq z_s. \label{join_exp}
	\end{align}
	Moreover, when $Z_m = Z_n$,
	\begin{align}
	\E[Z_m^2 | Z_s = z_s] = z_s p_m (1-p_m) + (z_s  p_m)^2. \label{joint_exp1}
	\end{align}
	
Additionally, from \eqref{meanx}, 
\begin{align}
\E[\hat{x}_0|Z_s = z_s] =\frac{1}{z_s}\sum_{m=1}^{M}z_s p_m x_m = \sum_{m=1}^M x_m p_m. \label{mean}
\end{align}
which is not a function of $z_s$. Therefore, 
\begin{align}
\E[\hat{x}_0] = \sum_{m=1}^M x_m p_m, \; \E[\hat{y}_0] = \sum_{m=1}^M y_m p_m. \label{mean1}
\end{align} 

Also, \eqref{join_exp}, \eqref{joint_exp1} and \eqref{mean} can be substituted into \eqref{cent} in order to evaluate the conditional mean-square error. Finally, the mean-square error of the centroid estimator is evaluated using \eqref{mse_cent}.  

\subsection{Asymptotically Unbiased Centroid (AUC) Estimator}
\begin{theorem}
 If the values of $I_0$ and $\lambda_n$ are known, $\rho$ is much smaller than the dimensions of the array $(\rho << a)$, and $(x_0, y_0)$ is within the bounds of the array, then an unbiased centroid estimator of the beam position can be realized in the limit as $M \to \infty$. The asymptotically unbiased  centroid estimator is defined as
\begin{align}
\hat{x}_0 \triangleq  \mathcal{K} \frac{1}{Z_s}\sum_{m=1}^M x_m Z_m, \quad  \hat{y}_0 \triangleq \mathcal{K} \frac{1}{Z_s}\sum_{m=1}^M y_m Z_m.
\end{align}
where $\mathcal{K} \triangleq \frac{\Lambda_s } {2 \pi I_0 }$. 
\end{theorem}
 \begin{proof}
 	Consider the mean value of the centroid estimator in \eqref{mean1}. The expectation of $\hat{x}_0$ can be further expanded as
\begin{align}
\E[\hat{x}_0] = \frac{1}{\Lambda_s} \sum_{m=1}^M x_m \iint_{A_m} \left( \frac{I_0}{\rho^2} \exp\left(-\frac{(x-x_0)^2 + (y-y_0)^2}{2 \rho^2}\right) + \lambda_n \right) \, dx \, dy.
\end{align}
In the limit as $M \to \infty$,
\begin{align}
&\E[\hat{x}_0] = \frac{1}{\Lambda_s} \iint_{\mathcal{A}} x\left( \frac{I_0}{\rho^2} \exp\left(-\frac{(x-x_0)^2 + (y-y_0)^2}{2 \rho^2}\right) + \lambda_n \right) \, dx \, dy\\
&= \frac{2 \pi I_0}{\Lambda_s} \int_{-a}^a \int_{-a}^a    x \frac{1}{2 \pi \rho^2} \exp\left(-\frac{(x-x_0)^2 + (y-y_0)^2}{2 \rho^2}\right) \, dx \, dy   + \frac{\lambda_n}{\Lambda_s} \cancelto{0}{ \int_{-a}^a \int_{-a}^a    x  \, dx \, dy}\\
& = \frac{2 \pi I_0}{\Lambda_s} x_0.
\end{align}
In a similar fashion, $\E[\hat{Y}_0] = \frac{2 \pi I_0}{\Lambda_s} y_0$.

Therefore, multiplying the (regular) centroid estimator by the factor $\mathcal{K} = \frac{\Lambda_s}{2 \pi I_0}$ results in an unbiased estimator as $M \to \infty$. 
\end{proof}

\subsubsection{High Signal-To-Noise Ratio Case}Additionally, it can also be observed that for high signal-to-noise ratio, the factor 
\begin{align}
\frac{2\pi I_0}{\Lambda_s}& = \frac{2\pi I_0}{\iint_{\mathcal{A}} \left( \frac{I_0}{\rho^2} \exp\left(-\frac{(x-x_0)^2 + (y-y_0)^2}{2 \rho^2}\right) + \lambda_n \right) \, dx \, dy} \approx \frac{2\pi I_0}{\iint_{\mathcal{A}} \left( \frac{I_0}{\rho^2} \exp\left(-\frac{(x-x_0)^2 + (y-y_0)^2}{2 \rho^2}\right) \right) \, dx \, dy} = 1.
\end{align}
Hence, $\E[\hat{x}_0] = x_0$, and the centroid estimate is asymptotically unbiased for high SNR.

Let $Z_s' \triangleq \frac{Z_s}{\mathcal{K}}$. Then, $z_s' = 0, \frac{1}{\mathcal{K}}, \frac{2}{\mathcal{K}}, \dotsc$, and $P(\{ Z_s' = z_s'\}) = P(\{ Z_s = \mathcal{K} z_s' \})$. Furthermore, it is straightforward to show that 
\begin{align}
P(\{ Z_m = z_m \} | \{ Z_s' = z_s' \}) = \binom{\mathcal{K}z_s'}{z_m} p_m^{z_m} (1-p_m)^{\mathcal{K}z_s' - z_m}, \quad m=1, \dotsc, M, \label{meanx1}
\end{align}
where  $0 \leq z_m \leq \mathcal{K}z_s'$. The expressions for $\E[Z_m^2 | Z_s' = z_s']$ and $\E[Z_m Z_n | Z_s' = z_s']$ can be obtained similarly by replacing $z_s$ with $\mathcal{K}z_s'$ in \eqref{mean1} and \eqref{mean}, respectively. Finally, the conditional mean-square error, $\E[(\hat{x}_0 - x_0)^2 + (\hat{y}_0 - y_0)^2| Z_s' = z_s']$, is computed by replacing $Z_s$ and $z_s$ with $Z_s'$ and $z_s'$, respectively, in \eqref{cent}. Finally, for the AUC estimator,
\begin{align}
\E[\hat{x}_0] = \mathcal{K}\sum_{m=1}^M x_m p_m, \; \E[\hat{y}_0] = \mathcal{K}\sum_{m=1}^M y_m p_m. \label{mean2}
\end{align}

	\subsection{Adaptive Centroid Estimators (ACE)}\label{ace}
	An adaptive centroid estimator takes the average of nonlinearly weighted photon counts. Such an estimator is designed to weight the photon counts for those detectors more heavily where the signal beam is expected to reside, i.e, where the photon count is relatively larger. These estimators are robust in the sense that they do not require the knowledge of beam parameters for estimation purpose, and for small $M$, provide a mean-square error performance which is better than the AUC estimator.
	
	\subsubsection{Adaptive Centroid Estimator 1}
	The Aaptive Centroid Estimator 1 is a function of a positive real number $n$. It is defined as, 
	\begin{align}
	\hat{x}_0 \triangleq  \frac{1}{Z_{s}}\sum_{m=1}^M x_m Z_m^n, \quad \hat{y}_0 \triangleq \frac{1}{Z_{s}}\sum_{m=1}^M y_m Z_m^n,
	\end{align}
	where $n \geq 1 $ and $Z_s \triangleq \sum_{m=1}^M Z_m^n$.  We note that when $n=1$, we obtain the centroid estimator, and when $n \to \infty$, the maximum detector count estimator is realized. 
	
	\subsubsection{Adaptive Centroid Estimator 2}
	The Adaptive Centroid Estimator 2 is a function of $n$ and $N$ where $n \geq 1$, and $N \in \mathbb{Z}^+, N < M$. In this case, we use the $N$ largest order statistics of the observations $Z_1, Z_2, \dotsc, Z_M$ for the centroid estimator. The estimator is defined as
	\begin{align}
	\hat{x}_0 \triangleq  \frac{1}{Z_{s}}\sum_{m=M-N+1}^M x_m \left(Z^{(m)} \right)^n,  \quad \hat{y}_0 \triangleq \frac{1}{Z_{s}}\sum_{m=M-N+1}^M y_m \left(Z^{(m)}\right)^n,
	\end{align}
	where $Z^{(1)} \leq Z^{(2)} \leq \dotsm \leq Z^{(M)}$ are the order statistics of $Z_1, \dotsc, Z_M$. Furthermore, $Z_s \triangleq \sum_{m=M-N+1}^M \left(Z^{(m)} \right)^n$.
			
	\paragraph{Asymptotic Behavior} \label{asymptotics}As $M \to \infty$, $A \to 0$, which implies that 
	\begin{align}
	&P(\{Z_m = 0\}) = \exp({-\Lambda_m}) \approx  \exp\left(-\left[\frac{I_0}{\rho^2}e^{-\frac{(x_m-x_0)^2 + (y_m-y_0)^2}{2 \rho^2}} + \lambda_n\right] A \right) \\ \nonumber
	&= 1 -  \left[\frac{I_0}{\rho^2}e^{-\frac{(x_m-x_0)^2 + (y_m-y_0)^2}{2 \rho^2}} + \lambda_n\right] A + o(A) \approx 1 -  \left[\frac{I_0}{\rho^2}e^{-\frac{(x_m-x_0)^2 + (y_m-y_0)^2}{2 \rho^2}} + \lambda_n\right] A \quad \forall m
	\end{align}
	where $o(A)$ is a function such that $\displaystyle \lim_{A \to 0} \frac{o(A)}{A}=0$. In a similar fashion, it can be shown (for small $A$) that $ \displaystyle P(\{Z_m = 1 \}) \approx \left[\frac{I_0}{\rho^2}e^{-\frac{(x_m-x_0)^2 + (y_m-y_0)^2}{2 \rho^2}} + \lambda_n\right] A$, and $P(\{ Z_m = \ell \}) \approx 0$ for any integer $\ell > 1$ for any $m$. This implies that ACE1 and ACE2 converge to the centroid estimator (or the AUC if they are scaled by $\mathcal{K}$) as $M \to \infty$.

  	\begin{figure}
		\centering
		\hspace{-1cm}
		\begin{subfigure}[t]{0.5\textwidth}
			\includegraphics[width=\textwidth]{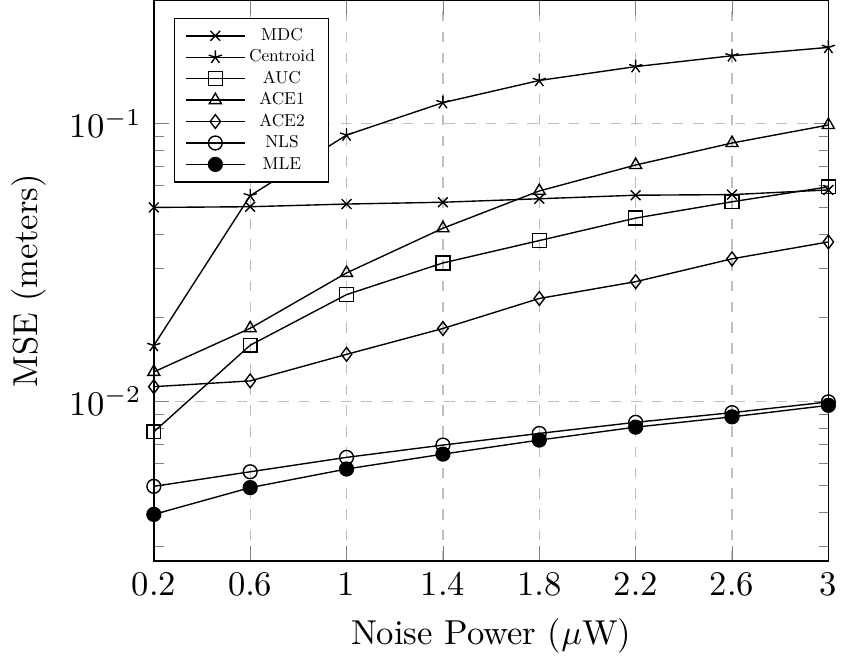} 
			\caption{MSE} \label{fig1a}
		\end{subfigure}
		\vspace{-3cm}
		\begin{subfigure}[t]{0.5\textwidth}
			\includegraphics[width=\textwidth]{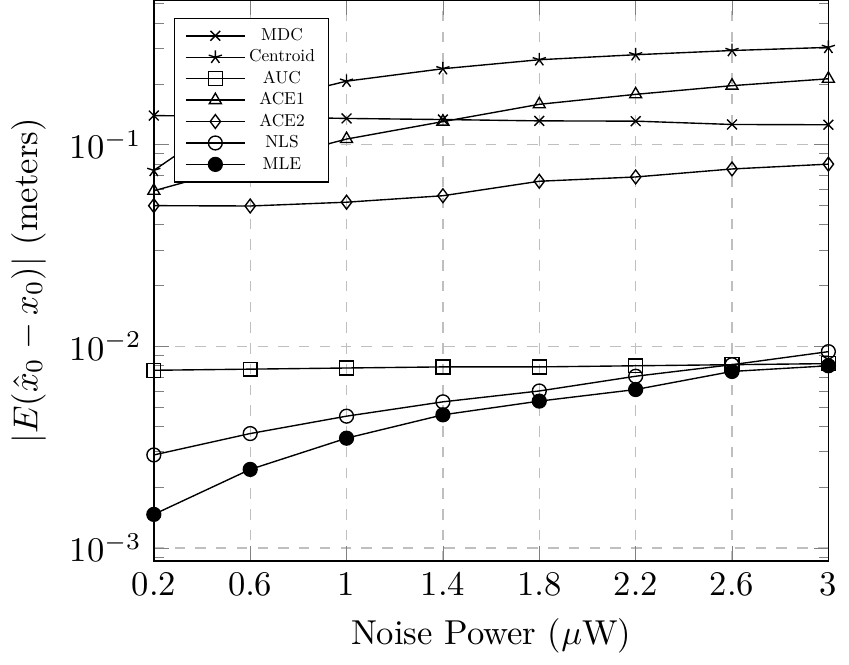} 
			\caption{Bias magnitude ($x$ axis)} \label{fig1b}
		\end{subfigure}
		\vspace{2.8cm}
		\caption{Fig.~\ref{fig1a} depicts the root mean-square error for different estimators of beam position for $4\times4$ detector array.  The signal power is 1 $\mu$W, and the noise power is varied between 0.2 and 1.8 $\mu$W. The beam was centered at $(0.4, 0.4)$,  and we used $n=2$ (second order) and $N=3$ for the ACE estimators. Fig.~\ref{fig1b} shows the magnitude of the bias plots along $x$ axis of different estimators.} \label{fig1}
	\end{figure}

\begin{figure}
	\centering
	\hspace{-1cm}
	\begin{subfigure}[t]{0.5\textwidth}
		\includegraphics[width=\textwidth]{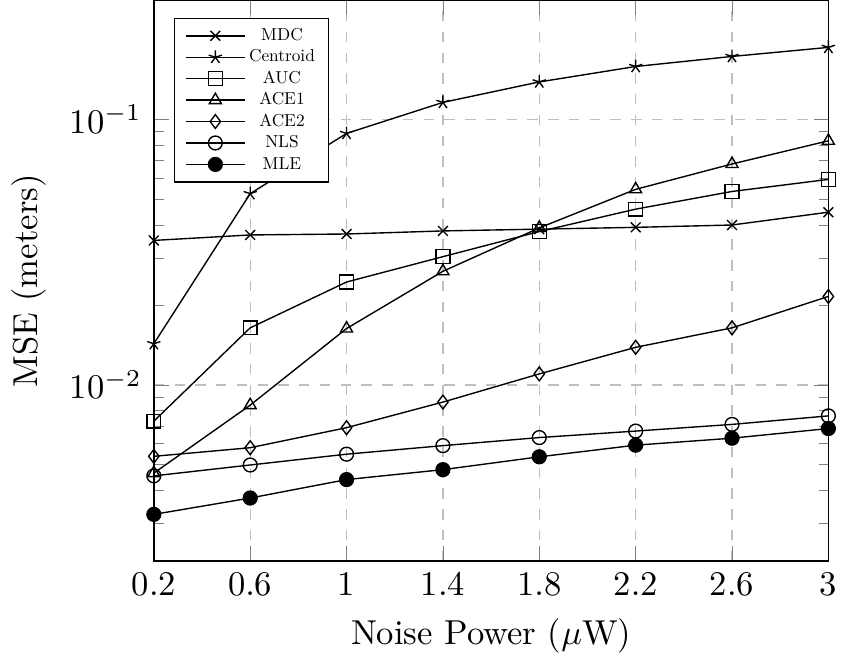} 
		\caption{MSE} \label{fig2a}
	\end{subfigure}
	\vspace{-3cm}
	\begin{subfigure}[t]{0.5\textwidth}
		\includegraphics[width=\textwidth]{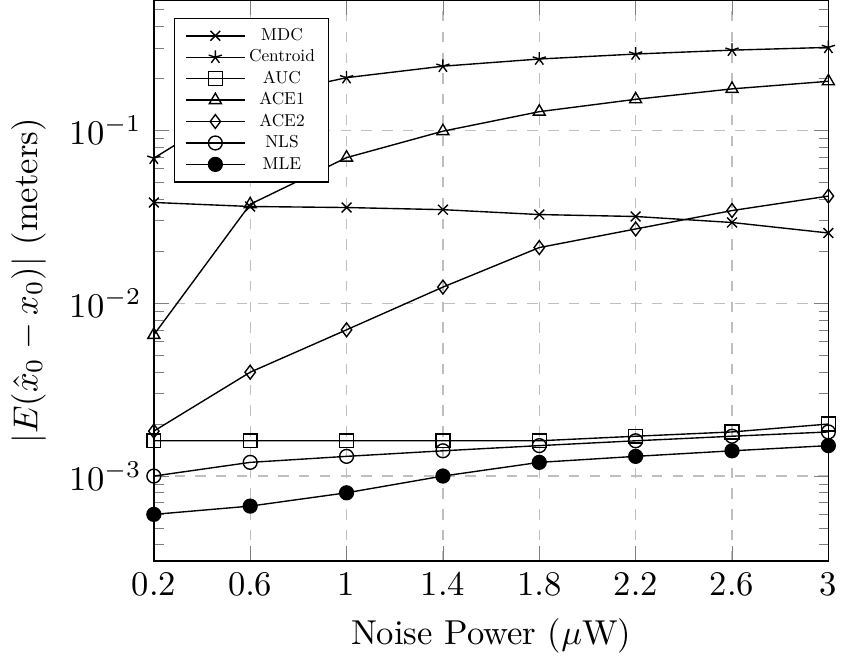} 
		\caption{Bias magnitude ($x$ axis)} \label{fig2b}
	\end{subfigure}
	\vspace{2.8cm}
	\caption{Fig.~\ref{fig2a} depicts the root mean-square error for different estimators of beam position for $6\times6$ detector array.  The signal power is 1 $\mu$W, and the noise power is varied between 0.2 and 1.8 $\mu$W. The beam was centered at $(0.4, 0.4)$,  and we used $n=2$ (second order) and $N=3$ for the ACE estimators. Fig.~\ref{fig2b} shows the magnitude of the bias plots along $x$ axis of different estimators.} \label{fig2}
\end{figure}

\begin{figure}
	\centering
	\hspace{-1cm}
	\begin{subfigure}[t]{0.5\textwidth}
		\includegraphics[width=\textwidth]{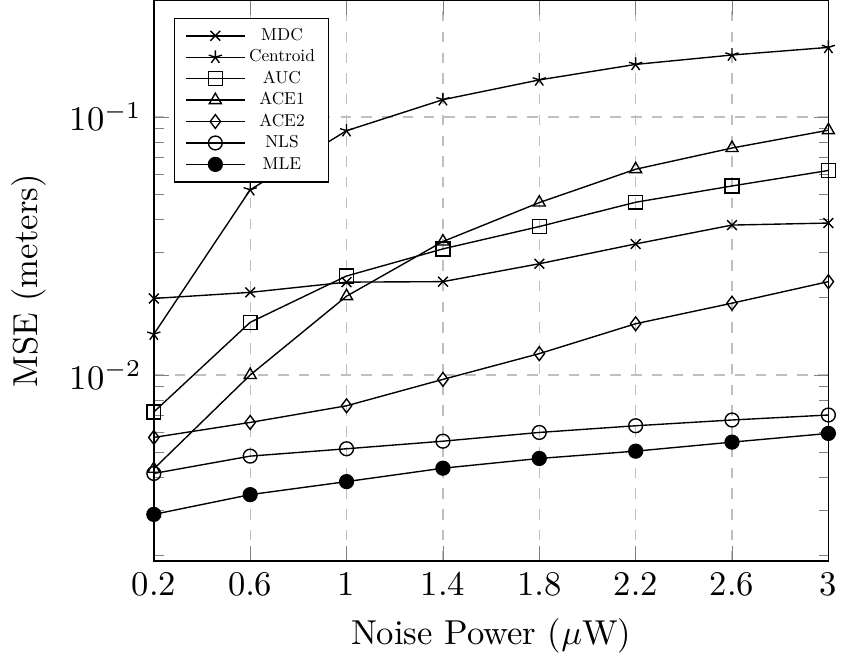} 
		\caption{MSE} \label{fig3a}
	\end{subfigure}
	\vspace{-3cm}
	\begin{subfigure}[t]{0.5\textwidth}
		\includegraphics[width=\textwidth]{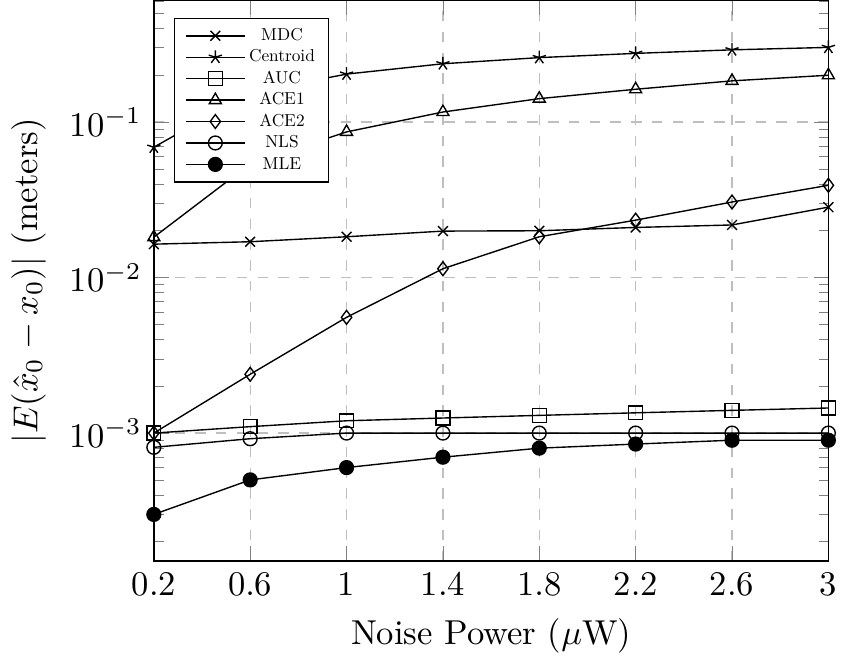} 
		\caption{Bias magnitude ($x$ axis)} \label{fig3b}
	\end{subfigure}
	\vspace{2.8cm}
	\caption{Fig.~\ref{fig3a} depicts the root mean-square error for different estimators of beam position for $8\times8$ detector array.  The signal power is 1 $\mu$W, and the noise power is varied between 0.2 and 1.8 $\mu$W. The beam was centered at $(0.4, 0.4)$,  and we used $n=2$ (second order) and $N=3$ for the ACE estimators. Fig.~\ref{fig3b} shows the magnitude of the bias plots along $x$ axis of different estimators.} \label{fig3}
\end{figure}

\begin{figure}
	\centering
	\hspace{-1cm}
	\begin{subfigure}[t]{0.5\textwidth}
		\includegraphics[width=\textwidth]{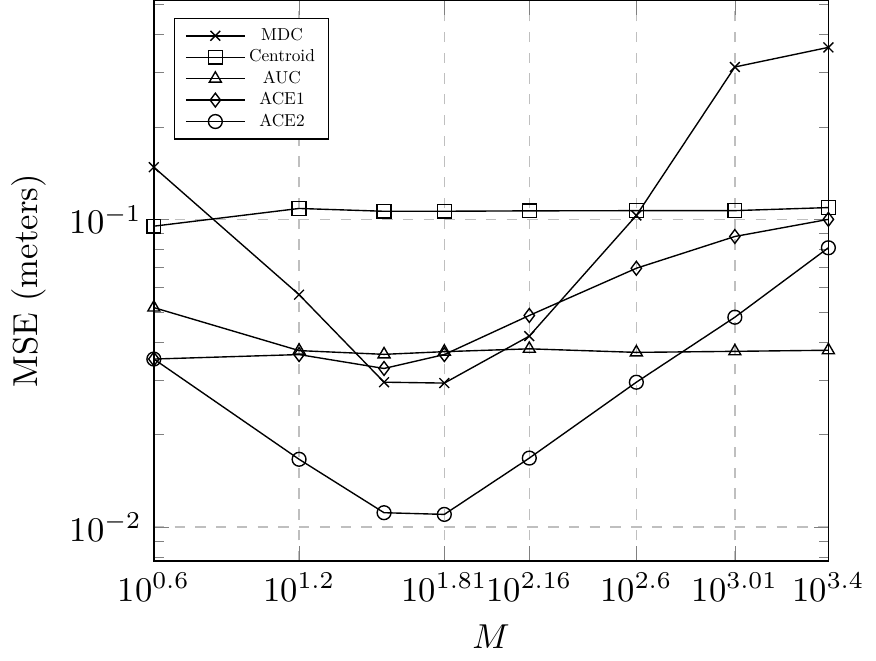} 
		\caption{Effect of $M$ on MSE} \label{fig5a}
	\end{subfigure}
	\vspace{-3cm}
	\begin{subfigure}[t]{0.5\textwidth}
		\includegraphics[width=\textwidth]{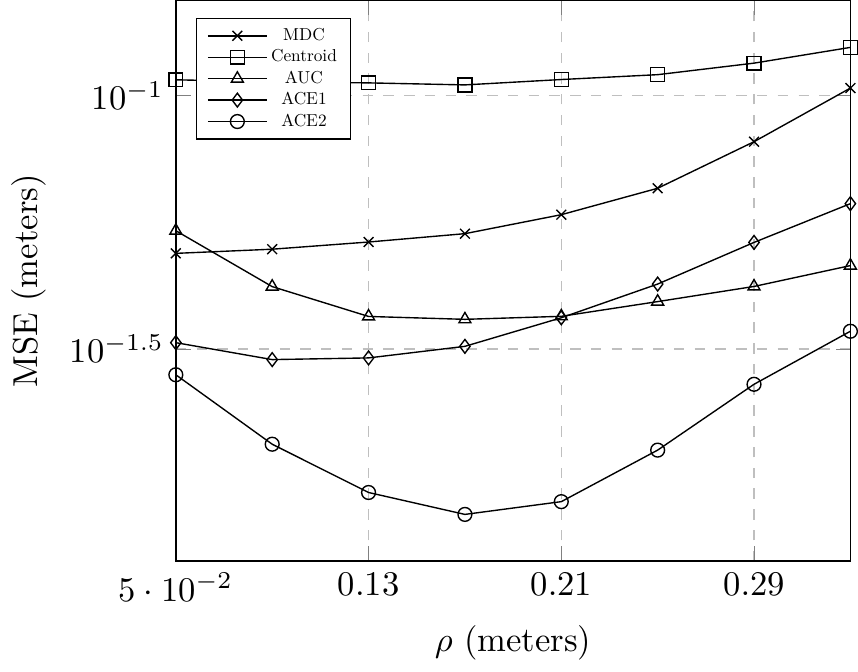} 
		\caption{Effect of $\rho$ on MSE} \label{fig5b}
	\end{subfigure}
	\vspace{2.8cm}
	\caption{Fig.~\ref{fig5a} depicts the effect of the number of detectors $M$ in the array on the MSE performance of different low complexity estimators. For this simulation, the signal power was fixed at 1 $\mu$W, noise power at 1.8 $\mu$W.   Fig.~\ref{fig5b} shows the MSE curves for different values of beam radius $\rho$ for similar signal and noise parameters and $M=16$. For both figures, $(x_0, y_0)$ was sampled randomly on the detector array.} \label{fig5}
\end{figure}

	\section{Probability of Error Performance}\label{PE}
In this section, we analyze the effect of beam position estimation on the probability of error for a PPM scheme. To this end, let us assume that we have a maximum likelihood receiver that operates on a symbol-by-symbol basis on a train of $\mathcal{M}$-PPM symbols. It is shown in \cite{Bashir3} that the probability of a correct decision, given a symbol $j$ is transmitted,  is
	\begin{align}
	P(c|j) = \left( P\left( \left\{ \sum_{m=1}^M \alpha_m Z_m^{(j)} - \sum_{m=1}^M \alpha_m Z_m^{(i)} > 0 \right\} \right) \right)^{\mathcal{M} - 1}, \label{mld}
	\end{align}
for $i, j = 1, 2, \dotsc, \mathcal{M}$ and $i \neq j$. The slot $j$ of the PPM symbol corresponds to the ``signal+noise'' slot, whereas $i$ corresponds to the ``noise only'' slot. The factor $\alpha_m$ is defined to be \cite{Bashir3}
\begin{align}
\alpha_m = \ln(1 + \text{SNR}_m),
\end{align}
where $\text{SNR}_m \triangleq \frac{\iint_{A_m}\frac{I_0}{\rho^2}e^{-\frac{(x-x_0)^2 + (y-y_0)^2}{2 \rho^2}}\, dx \, dy}{\lambda_n A}$ is the signal-to-noise ratio in the $m$th cell. Furthermore, $\E[Z_m^{(j)}] = \Lambda_m$ and $\E[Z_m^{(i)}] = \lambda_n A$. Let us define $Y_1 \triangleq \sum_{m=1}^M \alpha_m Z_m^{(j)}$ and $Y_0 \triangleq \sum_{m=1}^M Z_m^{(i)}$. By Gaussian approximation of a linear combination of Poisson random variables, both $Y_1$ and $Y_0$ are Gaussian random variables. Let $V\triangleq Y_1 - Y_0$. Then $V\sim \mathcal{N}\left(\mu_v, \sigma_v^2\right)$, where \begin{align}
\mu_v &\triangleq \sum_{m=1}^M \alpha_m \iint_{A_m}\frac{I_0}{\rho^2}e^{-\frac{(x-x_0)^2 + (y-y_0)^2}{2 \rho^2}}\, dx \, dy, \label{signal} \\
\sigma^2_v& \triangleq \sum_{m=1}^M \alpha_m^2 \iint_{A_m}\frac{I_0}{\rho^2}e^{-\frac{(x-x_0)^2 + (y-y_0)^2}{2 \rho^2}}\, dx \, dy + 2 \sum_{m=1}^M \alpha_m^2 \lambda_n A.\label{noise}
\end{align}
Thus, 
\begin{align}
P(c|j) \approx  \left( P\left( \{V > 0 \} \right) \right)^{\mathcal{M}-1}, \label{mld1}
\end{align}
 In order to maximize \eqref{mld1}, we need to maximize the factor $P(\{ V > 0\})$, which is given by
\begin{align}
P(\{V > 0 \}) = 1 - P(\{ V \leq 0\}) = 1 - \Phi\left( -\frac{\mu_v}{\sigma_v} \right). \label{mld3}
\end{align}
In order to maximize \eqref{mld3}, we need to maximize the factor $\frac{\mu_v}{\sigma_v}$ with respect to $(\hat{x}_0, \hat{y}_0)$. This factor is rewritten as
\begin{align}
\frac{\mu_v}{\sigma_v} = \frac{\sum_{m=1}^M \alpha_m \iint_{A_m}\frac{I_0}{\rho^2}e^{-\frac{(x-x_0)^2 + (y-y_0)^2}{2 \rho^2}}\, dx \, dy}{\sqrt{ \sum_{m=1}^M \alpha_m^2 \iint_{A_m}\frac{I_0}{\rho^2}e^{-\frac{(x-x_0)^2 + (y-y_0)^2}{2 \rho^2}}\, dx \, dy + 2 \sum_{m=1}^M \alpha_m^2 \lambda_n A}}.
\end{align}
When the SNR is low, $\lambda_n A >> \iint_{A_m}\frac{I_0}{\rho^2}e^{-\frac{(x-x_0)^2 + (y-y_0)^2}{2 \rho^2}}\, dx \, dy$ for each $m$. In this case, 
\begin{align}
\alpha_m = \ln\left( 1 + \frac{1}{\lambda_n A} \iint_{A_m}\frac{I_0}{\rho^2}e^{-\frac{(x-\hat{x}_0)^2 + (y-\hat{y}_0)^2}{2 \rho^2}}\, dx \, dy\right) \approx \frac{1}{\lambda_n A} {\iint_{A_m} \frac{I_0}{\rho^2}e^{-\frac{(x-\hat{x}_0)^2 + (y-\hat{y}_0)^2}{2 \rho^2}}\, dx \, dy} .
\end{align}
Additionally, we assume that $M$ is large, and let us denote the small area $A$ by $\Delta_M$. Then
\begin{small}
\begin{align}
&\frac{\mu_v}{\sigma_v}\nonumber \\
& \approx \! \frac{\frac{1}{\lambda_n \Delta_M}\sum_{m=1}^M  { \frac{I_0}{\rho^2}e^{-\frac{(x_m-\hat{x}_0)^2 + (y_m-\hat{y}_0)^2}{2 \rho^2}} \Delta_M}  \times \frac{I_0}{\rho^2}e^{-\frac{(x_m-x_0)^2 + (y_m-y_0)^2}{2 \rho^2}}\Delta_M}{\sqrt{ \sum_{m=1}^M \! \! \left(\frac{1}{\lambda_n \Delta_M} { \frac{I_0}{\rho^2}e^{-\frac{(x_m-\hat{x}_0)^2 + (y_m-\hat{y}_0)^2}{2 \rho^2}}\!\Delta_M} \!\right)^2\!\!\!\!\! \times\!  \frac{I_0}{\rho^2}e^{-\frac{(x_m-x_0)^2 + (y_m-y_0)^2}{2 \rho^2}}\!\Delta_M \!+ \!2  \lambda_n \Delta_M \! \sum_{m=1}^M \!\! \left(\!\frac{1}{\lambda_n \Delta_M} { \frac{I_0}{\rho^2}e^{-\frac{(x_m-\hat{x}_0)^2 + (y_m-\hat{y}_0)^2}{2 \rho^2}}\!\Delta_M} \!\!\right)^2 }}.
\end{align}
\end{small}
We assume that the search space for the maximization problem concerning $(\hat{x}_0, \hat{y}_0)$ is confined within $\mathcal{A}$, and that the assumption $2\pi \rho^2 << |\mathcal{A}|$ also holds. Then, the factor \begin{align}
&2  \lambda_n \Delta_M  \sum_{m=1}^M  \left(\frac{1}{\lambda_n \Delta_M} { \frac{I_0}{\rho^2}e^{-\frac{(x_m-\hat{x}_0)^2 + (y_m-\hat{y}_0)^2}{2 \rho^2}}\Delta_M} \right)^2 \approx
\frac{2}{\lambda_n} \iint_{\mathcal{A}} \left( \frac{I_0}{\rho^2} \right)^2 e^{-\frac{(x-\hat{x}_0)^2 + (y-\hat{y}_0)^2}{ \rho^2}}\, dx \, dy\nonumber \\
&= \frac{2 I_0^2}{\lambda_n \rho^4} 2 \pi \left(\frac{\rho}{\sqrt{2}}\right)^2 \iint_{\mathcal{A}} \frac{1}{2 \pi \left(\frac{\rho}{\sqrt{2}}\right)^2}  e^{-\frac{(x-\hat{x}_0)^2 + (y-\hat{y}_0)^2}{2 (\rho/\sqrt{2})^2}}\, dx \, dy = \frac{2 \pi I_0^2}{\lambda_n \rho^2}.
\end{align}
By further simplifications, it follows that
\begin{align}
\frac{\mu_v}{\sigma_v} \approx \frac{\frac{1}{\lambda_n} \iint_{\mathcal{A}} \left( \frac{I_0}{\rho^2}\right)^2 e^{-\frac{(x-\hat{x}_0)^2 + (y-\hat{y}_0)^2}{2 \rho^2}} e^{-\frac{(x-x_0)^2 + (y-y_0)^2}{2 \rho^2}} \, dx \, dy }   { \sqrt{ \frac{1}{\lambda_n^2} \iint_{\mathcal{A}} \left( \frac{I_0}{\rho^2}\right)^3 e^{-\frac{(x-\hat{x}_0)^2 + (y-\hat{y}_0)^2}{ \rho^2}} e^{-\frac{(x-x_0)^2 + (y-y_0)^2}{2 \rho^2}} \, dx \, dy + \frac{2 \pi I_0^2}{\lambda_n \rho^2} } }.\label{obj_fn}
\end{align}
The function given in \eqref{obj_fn} has to be maximized with respect to $(\hat{x}_0, \hat{y}_0)$. 
We first optimize with respect to $\hat{x}_0$, and due to the symmetric nature of the Gaussian beam, the same analysis will hold for optimization with respect to $\hat{y}_0$ as well. Thus, taking the natural log of \eqref{obj_fn} and setting the derivative (with respect to $\hat{x}_0$) of that to zero,  we have that,
\begin{small}
\begin{align}
&\frac{\partial \ln\left( \mu_v / \sigma_v\right)}{\partial \hat{x}_0}\nonumber \\
&= \frac{\iint_{\mathcal{A}} e^{-\frac{(x-\hat{x}_0)^2 + (y-\hat{y}_0)^2}{2 \rho^2}} \left(\frac{x-\hat{x}_0}{\rho^2} \right) e^{-\frac{(x-x_0)^2 + (y-y_0)^2}{2 \rho^2}} \, dx \, dy}{\iint_{\mathcal{A}} e^{-\frac{(x-\hat{x}_0)^2 + (y-\hat{y}_0)^2}{2 \rho^2}} e^{-\frac{(x-x_0)^2 + (y-y_0)^2}{2 \rho^2}} \, dx \, dy} -  \frac{\frac{1}{\lambda_n^2} \iint_{\mathcal{A}} \left( \frac{I_0}{\rho^2}\right)^3 e^{-\frac{(x-\hat{x}_0)^2 + (y-\hat{y}_0)^2}{ \rho^2}} \left( \frac{ x-\hat{x}_0}{\rho^2}\right) e^{-\frac{(x-x_0)^2 + (y-y_0)^2}{2 \rho^2}} \, dx \, dy}{\frac{1}{\lambda_n^2} \iint_{\mathcal{A}} \left( \frac{I_0}{\rho^2}\right)^3 e^{-\frac{(x-\hat{x}_0)^2 + (y-\hat{y}_0)^2}{ \rho^2}}  e^{-\frac{(x-x_0)^2 + (y-y_0)^2}{2 \rho^2}} \, dx \, dy + \frac{2 \pi I_0^2}{\lambda_n \rho^2}} \label{opt1} \\
&= 0. \label{opt}
\end{align}
\end{small}
It can be easily seen that $\hat{x}_0 = x_0$ causes both terms in \eqref{opt1} to go to zero, and therefore, this particular point is a solution. Hence, $\hat{x}_0 = x_0$ is a critical point (a local or global minimizer or maximizer). Similarly, $\hat{y}_0 = y_0$ is also a critical point. Additionally, it can be further shown that the Hessian matrix of $\ln(\mu_v/\sigma_v)$ is negative definite for all values of $(\hat{x}_0, \hat{y_0})$ and $(x_0, y_0)$ inside $\mathcal{A}$ when $I_0, \rho$ and $\lambda_n$ are positive. Hence, $\ln(\mu_v / \sigma_v)$ is concave with respect to $(\hat{x}_0, \hat{y}_0)$, and $(\hat{x}_0, \hat{y}_0) = (x_0, y_0)$ is a global maximizer of $\ln(\mu_v/\sigma_v)$ in $\mathcal{A}$.

\begin{figure}
	\centering
	\begin{subfigure}[t]{0.45\textwidth}
		\includegraphics[width=\textwidth]{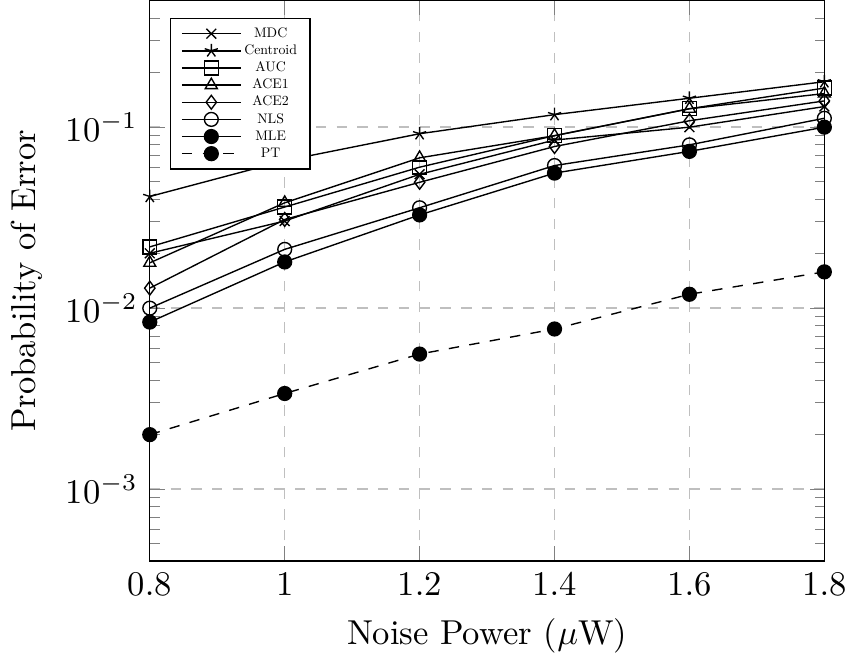} 
		\caption{ $M=16$} \label{fig6a}
	\end{subfigure}
	
	\hspace{0cm}
	\begin{subfigure}[t]{0.45\textwidth}
		\includegraphics[width=\textwidth]{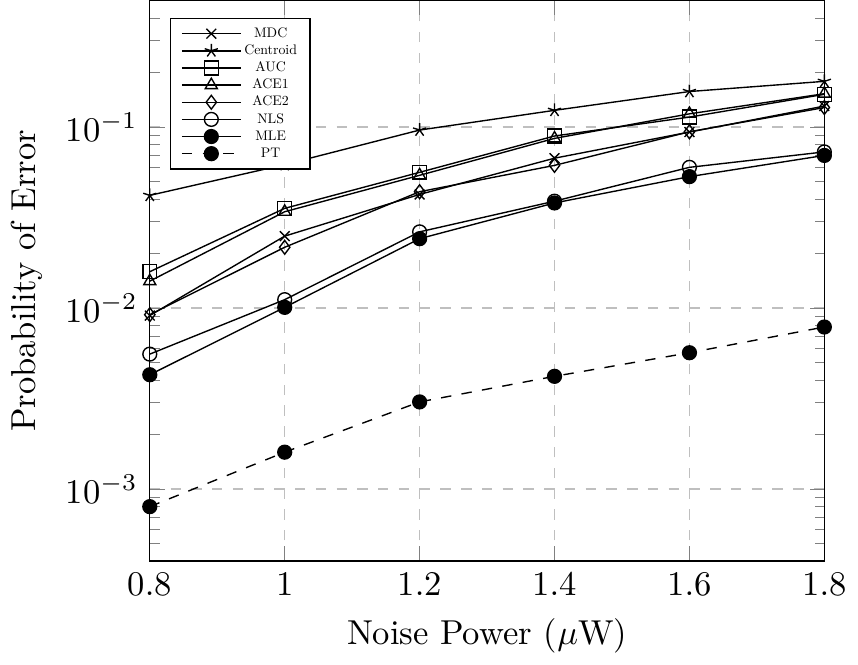} 
		\caption{ $M=36$} \label{fig6b}
	\end{subfigure}
	\vspace{0cm}
	\begin{subfigure}[t]{0.45\textwidth}
		\includegraphics[width=\textwidth]{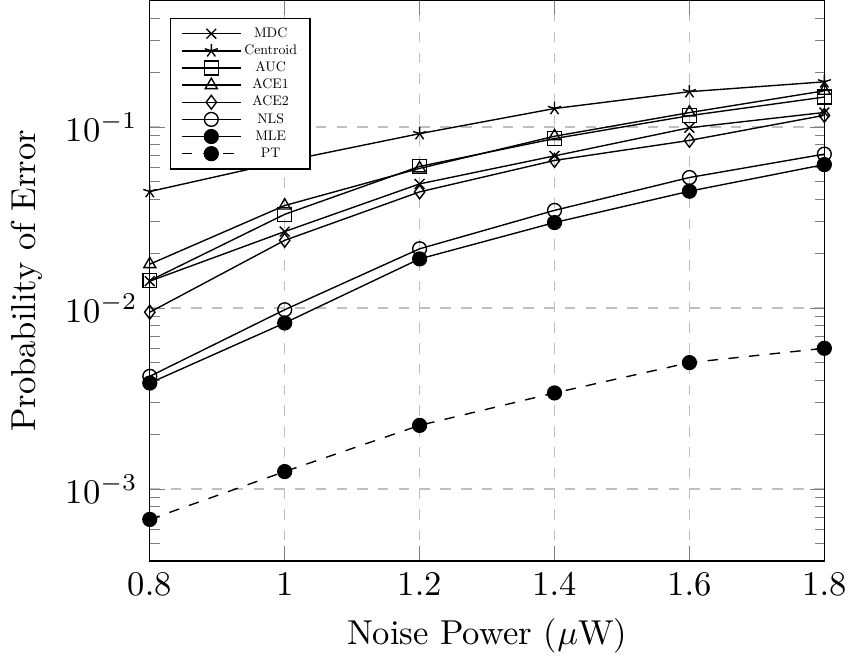} 
		\caption{$M=64$} \label{fig6c}
	\end{subfigure}
	\vspace{2.8cm}
	\caption{Fig.~\ref{fig6a} depicts the probability of error when $M=16$, Fig.~\ref{fig6b} for $M=36$, and Fig.~\ref{fig6c} presents the probability of error when $M=64$ for different estimators. For this simulation, the signal power was fixed at 0.5 $\mu$W, and the noise power was varied from 0.8 $\mu$W to 1.8 $\mu$W.   The beam radius $\rho$ was set at 0.2 meter.  For all figures, $(x_0, y_0)$ was sampled randomly on the detector array. The value of $|\mathcal{A}|$ is 4 square meters.} \label{fig6}
\end{figure}

\section{Simulation Preliminaries and Results} \label{sims}
For the mean-square error curves, we used \eqref{mse} and \eqref{mse_cent} to plot the mean-square error of maximum detector count estimator and centroid/AUC estimators\footnote{In case of AUC, we have to replace $Z_s$ with $Z_s'$ as mentioned in the section on AUC.}, respectively. Similarly, the equations \eqref{bias_mdc}, \eqref{mean1} and \eqref{mean2} are used to plot the bias functions of the MDC, centroid and AUC estimators, respectively. The quantity $P_{n,m}$ represents an infinite sum in \eqref{mse} and \eqref{bias_mdc}. Similarly, \eqref{mse_cent} also represents an infinite sum with respect to $z_s$. 

Let us consider the factor $P_{n,m}$ for $n=1, 2, \dotsc, M$ and $m$ fixed. Then, for the purpose of simulations, the infinite sum in the computation of $P_{n,m}$ is approximated by a finite sum. The upper limit in the sum is replaced by a large number $\eta_m$ such that the sum $S_m$ over the distribution
\begin{align}
S_m \triangleq \sum_{z_m=\eta_m}^\infty e^{-\Lambda_u} \frac{\Lambda_u^{z_m}}{z_m!} < \frac{\epsilon_0}{M!},
\end{align} 
where $\Lambda_u \triangleq \max(\Lambda_1, \Lambda_2, \dotsc, \Lambda_M)$, and $\eta_m >> \Lambda_u$. Therefore, we can replace the $\infty$ in the infinite sum of $P_{n,m}$ by $\eta_m$, and the error in approximation of $P_{n,m}$ is guaranteed to be less than $\epsilon_0$ for any $n$ and $m$. The $(M-1)!$ term appears because there are at least $(M-1)!$ terms over which the infinite sum is computed for any $P_{k,m}$, $k=1, \dotsc, M$ (please see \eqref{MDC}). The careful reader will discern that for a fixed $\epsilon$ and fixed $M$, both $\Lambda_u$ and $\eta_m$ are functions of $(x_0, y_0)$ if all other parameters of the beam are fixed. For the purpose of our simulations, we chose $\eta_m$ so that $\epsilon_0  \leq 10^{-5}$ for each $m$. 

By using the same arguments as before, the upper limit in the infinite sum is replaced by $\eta$. For the purpose of simulations, we chose $\eta$ so that the sum 
$S \triangleq \sum_{z_s=\eta}^\infty e^{-\Lambda_s} \frac{\Lambda_s^{z_s}}{z_s!} < \epsilon \leq 10^{-5}.$

For the rest of the estimators, it is not straightforward to come up with analytical expressions for the mean-square error and bias functions. Therefore, we resorted to Monte Carlo simulations in order to estimate the error and bias values.

First, we analyzed the mean-square error and the absolute bias values of different proposed estimators for a low signal-to-noise ratio regime, which is typically the scenario of interest. Fig.~\ref{fig1}, Fig.~\ref{fig2} and Fig.~\ref{fig3} indicate that the MLE and NLS estimators perform better than the low complexity estimators in terms of both the criteria.  We also note that the AUC estimator's bias diminishes with $M$.

Fig.~\ref{fig5} shows the effect of the number of cells $M$ and the radius $\rho$ on the mean-square error performance of different estimators. The two estimators, NLS estimator and MLE, were not included in this analysis because their computational complexity becomes prohibitively expensive for the purpose of simulations\footnote{A small $\rho$ incurs a large sharp peak for the likelihood function inside the cell where the true beam center is located, and the same fact holds for the objective function of the NLS estimator. This increases the time it takes for the evolutionary algorithm to converge to the maximizer. Additionally, a large $M$ increases the number of terms in the NLS objective function and the likelihood function, thereby causing the complexity to grow.}. It is interesting to note that the value of $M$ does not have a significant effect on  the performance of the centroid and AUC estimators. This is explained by the fact that for $x_0$, the centroid estimator is an averaging estimator: $\hat{x}_0 \triangleq \frac{1}{Z_s} \bZ^T \bx$, where $\bx \triangleq \begin{bmatrix}
x_1 & x_2 & \dots & x_M
\end{bmatrix}^T$, $x_m$ being the center of the $m$th cell. Since each of the coordinates $x_m$ are weighted linearly, then it can be seen, at least intuitively, that $\hat{x}_0$ will not change significantly if $M$ is increased for a fixed SNR. The same argument applies for $\hat{y}_0$, and the approximate constancy of AUC may be explained in a similar fashion. In the case of the rest of the estimators, the location coordinates are weighted nonlinearly, and hence, the mean-square error performance varies with $M$.

 Interestingly, as we had predicted in Paragraph~\ref{asymptotics}, the performance of ACE1 and ACE2 converges to the centroid estimator's performance as $M \to \infty$. Finally, we note that since ACE2 only utilizes the readings of top four photon counts in these simulations, the performance of ACE1 and ACE2 is exactly the same when $M=4$. 

The results obtained of Fig.~\ref{fig5b} can be explained in the light of arguments used to justify the behavior of the CRLB in Fig.~\ref{fig4b}. The interested reader is referred to Section~\ref{CRLB_analysis} in this regard. 

Fig.~\ref{fig6} presents the probability of error performance obtained with different beam position estimators/detector arrays. The curves are lower bounded by an ideal system's performance that essentially has perfect knowledge of the beam position. As can be observed, the lower bound becomes smaller with the number of detectors$M$ in the array. 

\section{A Brief Commentary on Computational Complexity of Tracking Algorithms}\label{Complexity}
Roughly, the complexity of low complexity estimators is on the order of $N^2$ real summations and $N^2$ real multiplies---barring the ACE2 and MDC---where $N = \sqrt{M}$ is the number of detectors along one side of the square shaped detector array. In case of MDC, we have to use a sorting technique in order to find the maximum out of an array of $N^2$ elements. For ACE2, we have approximately $L$ real sums and $L$ real multiplies in addition to the sorting algorithms complexity which sorts the highest $L$ numbers out of an array of $M$ numbers. We note again that $L=4$ for the purpose of simulations. Therefore, approximately, the complexity is a function of $N^2$ for low complexity trackers. 

The complexity of NLS and MLE is much higher: In addition to computing approximately $N^2$ real additions and $N^2$ real multiplies (see equations \eqref{estimate} and \eqref{mle}), the algorithms resort to a real number genetic algorithm in order to find the global maximum. The complexity of the real number genetic algorithm is discussed in \cite{Bashir1}. The complexity of the genetic algorithm is a function of number of chromosomes, $N_c$, and the number of generations\footnote{The number of generations can be regarded as the number of iterations required in order to converge to the true maximum/minimum of the objective function.}, $N_g$. The values of $N_c$ and $N_g$ should be chosen according to the nature of the objective function---a ``spikier'' function requires relatively large $N_g$ and $N_c$ for convergence to the true maximum. In our simulations, we set $N_c = 50$, and $N_g=400$. For each chromosome, the objective functions in given by \eqref{estimate} and \eqref{mle} are determined. Thus, the total complexity for the NLS or MLE tracking is approximately $N_c \times N_g \times N^2$ real multiplications and real additions \footnote{This does not include the complexity involved in comparing the fitness of the chromosomes during each iteration of the algorithm}.

The interested reader is referred to the excellent text \cite{Rao} for more details on genetic algorithms.

	\section{Conclusion}\label{Conc}
	
	In this paper, we have analyzed the problem of tracking with a photon-counting detector array receiver and Gaussian beams in a free-space optical communications system. From purely a communication theory point-of-view, an array of detectors is more useful from a single detector from two perspectives: i) The array of detectors minimizes the tracking error, and ii) and the detector arrays provide a better probability of error performance \cite{Bashir4}. However, through a study of the Cram\`er-Rao Lower Bound of the tracking error, we discovered that improvement in performance becomes smaller if we increase the number of cell from $M =N^2$ to $M = (N+1)^2$ when $N$ is large (law of diminishing returns). Moreover, the same law of diminishing returns applies to the probability of error performance as well. Additionally, the computational complexity increases linearly with the number of detectors $M$ in the array, and the circuit and storage complexities grow with $M$ as well. 
	
	Additionally, we also observed that the beam location on the array is not only required for tracking, but is also part of the channel state information required to decode pulse position modulation/on-off keying symbols. To this end, we proposed a number of non-Bayesian tracking algorithms and analyzed their mean-square error/probability of error performance. The two algorithms, namely the nonlinear least squares estimator and the maximum likelihood estimator, performed better than the different versions of the centroid algorithms. However, the two aforementioned algorithms incur a higher cost in terms of computational complexity. Additionally, since the computational overhead of maximum likelihood and nonlinear least squares estimators is comparable, the better performance of maximum likelihood estimator makes it a more viable tracking algorithm of the two. 
	
	 Therefore, depending on the trade-off between the performance and the price we are willing to pay for it, we can choose a certain number of detectors in our array, and a particular tracking algorithm, to track the beam position. Thus, if we are willing to invest in a more complex receiver in terms of an array of detectors and a high complexity algorithm like maximum likelihood, we can achieve better performance gains not only in terms of tracking, but also from the perspective of probability of error. 
	
	
	\bibliography{reference1.bib}

	\bibliographystyle{IEEEtran}
	
\end{document}